\documentclass{article}
\usepackage[utf8]{inputenc}
\usepackage[american]{babel}
\usepackage{enumerate,balance}
\usepackage{amsfonts,amsmath,amssymb,amsthm,amstext,latexsym}
\usepackage{paralist}
\usepackage{url}
\usepackage{xspace}
\usepackage[ruled,lined,boxed,commentsnumbered]{algorithm2e}
\usepackage{tikz}
\usepackage{array}
\usepackage{verbatim}
\usepackage{times}
\usepackage[shadow]{todonotes}
\usepackage{fullpage}

\usepackage{subcaption}

\usepackage{tikz}
\usetikzlibrary{positioning}

\newcommand{\ema}[1]{\ensuremath{#1}}
\newcommand{\kinp}[1][k]{\textsc{\ema{#1}-in-\ema{p}-Co\-Sche\-dule}\xspace}
\newcommand{\pinp}[1][p]{\textsc{\ema{#1}-in-\ema{#1}-Co\-Sche\-dule}\xspace}

\newcommand{\tpart}{3-\textsc{Partition}\xspace}

\newcommand{\leqb}{\ema{\preccurlyeq_{\sigma}}\xspace}
\newcommand{\bs}{1-pack-schedule\xspace}

\newcommand{\pack}{pack\xspace}
\newcommand{\packs}{packs\xspace}
\newcommand{\work}{work\xspace}
\newcommand\II{\ema{\mathcal{I}}\xspace}
\newcommand{\MS}{\ema{\text{COST}}\xspace}
\newcommand{\opt}{\ema{\textsc{opt}}\xspace}

\newcommand{\makebatch}{\textsc{Make-\pack}\xspace}

\newcommand{\ok}{\textsc{\pack-Approx}\xspace}
\newcommand{\bbb}{\textsc{\pack-by-\pack}\xspace}

\newcommand{\randompack}{\textsc{Random-Pack}\xspace}
\newcommand{\randomproc}{\textsc{Random-Proc}\xspace}
\newcommand{\email}[1]{\url{#1}}

\newcommand{\si}[1]{\ema{\sigma^{(#1)}}}

\newtheorem{lemma}{Lemma}
\newtheorem{theorem}{Theorem}

\begin{document}


\title{Co-Scheduling Algorithms for High-Throughput \\ Workload Execution}
\date{\today}

%
%

\author{
Guillaume Aupy\thanks{LIP, Ecole Normale Sup\'erieure de Lyon, France} \and Manu Shantharam\thanks{University of Utah, USA} \and Anne Benoit\footnotemark[1]~\thanks{Institut Universitaire de France, France} \and Yves Robert\footnotemark[1]\footnotemark[3]~\thanks{University of Tennessee Knoxville, USA} \and Padma Raghavan\thanks{Pennsylvania State University, USA}\\ 
\{guillaume.aupy, anne.benoit, yves.robert\}@ens-lyon.fr; \\
shantharam.manu@gmail.com; raghavan@cse.psu.edu
}

\maketitle

\begin{abstract}
This paper investigates co-scheduling algorithms for processing a set of parallel applications. 
Instead of executing each application one by one, using a maximum degree of parallelism 
for each of them, we aim at scheduling several applications concurrently. 
We partition the original application set into a series of \packs, which are executed one by one. 
A pack comprises several applications, each of them with an assigned number of 
processors, with the constraint that the total number of processors assigned within a \pack does not exceed
the maximum number of available processors. The objective is to determine a partition into \packs, and an assignment of processors to applications, that minimize the sum of the execution times of the \packs.
We thoroughly study the complexity of this optimization problem, and propose several heuristics that
exhibit very good performance on a variety of workloads, whose application execution times model 
profiles of  parallel scientific codes. We show that  co-scheduling leads to 
to faster workload completion time and to faster response times on average (hence increasing system throughput and saving energy), for significant benefits over traditional scheduling  from both the user and system perspectives.
\end{abstract}

\section{Introduction}

The execution time of many high-performance computing applications can be significantly reduced when using a large number of processors. Indeed, parallel multicore platforms enable the fast processing of 
very large size jobs, thereby rendering the solution of challenging scientific problems more tractable. 
However, monopolizing all computing resources to accelerate the processing of a single application
is very likely to lead to inefficient resource usage. This is because the typical speed-up profile of most applications is sub-linear and even reaches a threshold: when the number of processors increases,
the execution time first decreases, but not linearly, because it suffers from the overhead due to communications and load imbalance; at some point, adding more resources does not lead to any significant benefit.

In this paper, we consider a pool of several applications that have been submitted for execution. 
Rather than executing each of them in sequence, with the maximum number of available resources,
we introduce co-scheduling algorithms that execute several applications concurrently.
We do increase the individual
execution time of each application, but (i) we improve the efficiency of 
the parallelization, because each application is scheduled on fewer resources;
(ii) the total execution time will be much shorter; and (iii) the average response time
will also be shorter. In other words, co-scheduling
increases platform yield (thereby saving energy) without sacrificing
response time. 

In operating high performance computing systems, the costs of energy
consumption can greatly impact the total costs of ownership. Consequently,
there is a move away from a focus on peak performance (or speed) and
towards improving energy efficiency~\cite{shalf,balaji}. Recent results
on improving the energy efficiency of workloads can be broadly
classified into approaches that focus on dynamic
voltage and frequency scaling, or alternatively, task aggregation or
co-scheduling.
In both types of approaches, the individual execution time of
an application may increase but there can be considerable energy savings
in processing a workload.

More formally, we deal with the following problem: given (i) a distributed-memory platform 
with $p$ processors, and (ii) $n$ applications, or tasks, $T_{i}$, with their execution profiles
($t_{i,j}$ is the execution time of $T_{i}$ with $j$ processors), what is the best way to \emph{co-schedule} them, i.e., to partition them into \packs,  so as to minimize the sum of the execution times over all \packs.
Here a \pack is a subset of tasks, together with a processor assignment for each task. The constraint is
that the total number of resources assigned to the \pack does not exceed $p$, and the execution time of
the \pack is the longest execution time of a task within that \pack.
The objective of this paper is to study this co-scheduling problem, both theoretically and experimentally,
We aim at demonstrating the gain that can be achieved through co-scheduling, both on platform yield and response time, using a set of real-life application profiles. 

On the theoretical side, to the best of our knowledge, the complexity of the co-scheduling problem  
has never been investigated,
except for the simple case when one enforces that each \pack comprises at most $k=2$ tasks~\cite{syp}. 
While the problem has polynomial complexity for the latter restriction (with at most $k=2$ tasks per \pack), we show that it is NP-complete when assuming at most $k \geq 3$ tasks per pack. Note that the instance with
$k=p$ is the general, unconstrained, instance of the co-scheduling problem. We also propose an approximation algorithm for the general instance. In addition, we propose an optimal processor assignment procedure when the tasks that form a \pack are given. We use these two results to derive efficient heuristics.
Finally, we discuss how to optimally solve small-size instances, either through enumerating partitions, or through an integer linear program: this has a potentially exponential cost, but allows us to assess the absolute quality of the heuristics that we have designed. Altogether, all these results lay solid theoretical foundations
for the problem.

On the experimental side, we study the performance of the heuristics 
on a variety of workloads, whose application execution times model 
profiles of parallel scientific codes. We focus on three criteria:
(i) cost of the co-schedule, i.e., total execution time; 
(ii) packing ratio, which evaluates the idle time of processors
during execution; and (iii) response time compared to a
fully parallel execution of each task starting from shortest task. 
The proposed heuristics show very good performance within
a short running time, hence validating the approach. 


The paper is organized as follows. We discuss related work in Section~\ref{sec.related}. 
The problem is then formally defined in Section~\ref{sec.pb}. Theoretical results
are presented in Section~\ref{sec.theory}, exhibiting the problem complexity, 
discussing sub-problems and optimal solutions, and providing an approximation 
algorithm. Building upon these results, several polynomial-time heuristics
are described in Section~\ref{sec.heuristics}, and they are thoroughly evaluated
in Section~\ref{sec:exp}. Finally we conclude and discuss future work in Section~\ref{sec.conc}.

\section{Related work}
\label{sec.related}

In this paper, we deal with \pack scheduling for parallel tasks, aiming at makespan minimization
(recall that the makespan is the total execution time). The 
corresponding problem with sequential tasks (tasks that execute on a single processor) is easy to 
solve for the makespan minimization objective: simply make a \pack out of the largest $p$ tasks, and 
proceed likewise while there remain tasks. Note that the \pack scheduling problem with sequential tasks 
has been widely studied for other objective functions, see  Brucker et 
al.~\cite{brucker1997scheduling} for various job cost functions, and Potts and 
Kovalyov~\cite{potts2000scheduling} for a survey. 
Back to the problem with sequential tasks and the  makespan objective, Koole and Righter in~\cite{koole2001stochastic} 
deal with the case where the execution time of each task is unknown but defined by a probabilistic distribution. They 
showed counter-intuitive properties, that enabled them to derive an algorithm that computes the 
optimal policy when there are two processors, improving the result of Deb and 
Serfozo~\cite{deb1973optimal}, who considered the stochastic problem with identical jobs.

To the best of our knowledge, the problem with parallel tasks has not been studied as such. However, 
it was introduced by Dutot et al. in~\cite{dutot2003scheduling} as a moldable-by-phase model to 
approximate the moldable problem. The moldable task model is similar to the \pack-scheduling 
model, but one does not have the additional constraint (\pack constraint) that the execution of new tasks cannot start before all tasks in the current \pack are completed.
Dutot et al. in~\cite{dutot2003scheduling} provide an optimal polynomial-time solution for the problem of \pack scheduling
identical independent tasks, using a dynamic programming algorithm. This is the only 
instance of \pack-scheduling with parallel tasks that we found in the literature.

A closely related problem is the rectangle packing problem, or 2D-Strip-packing. Given a set of 
rectangles of different sizes, the problem consists in packing these rectangles into another rectangle of 
size $p\times m$. If one sees one dimension ($p$) as the number of processors, and the 
other dimension ($m$) as the maximum makespan allowed, this problem is identical to the variant of our problem
where the number of processors is pre-assigned to each task: each rectangle $r_i$ of size 
$p_i \times m_i$ that has to be packed can be seen as the task $T_i$ to be computed on 
$p_i$ processors, with $t_{i,p_i} = m_i$. In~\cite{turek1994scheduling}, Turek et al. approximated
the rectangle packing problem using \emph{shelf-based} solutions: the rectangles are assigned to 
\emph{shelves}, whose placements correspond to constant time values. All rectangles assigned to 
a shelf have equal starting times, and the next shelf is placed on top of the previous shelf. This is 
exactly what we ask in our \pack-scheduling model. 
This problem is also called level packing in some papers, and we refer the reader to a recent survey on 
2D-packing algorithms by Lodi et al.~\cite{lodi2002two}. In particular, Coffman et al. 
in~\cite{coffman1980performance} show that level packing algorithm can reach a $2.7$ 
approximation for the 2D-Strip-packing problem ($1.7$ when the length of each rectangle is bounded 
by 1).
Unfortunately, all these algorithms consider the number of processors (or width of the rectangles) to
be already fixed for each task, hence they cannot be used directly in our problem for which a key decision is to decide the number
of processors assigned to each task.

In practice, \pack scheduling is really
useful as shown by recent results.  Li et al.~\cite{li2010power} propose
a framework to predict the energy and performance impacts of power-aware
MPI task aggregation.  Frachtenberg et al.~\cite{fcs} show that system
utilization can be improved through their  schemes to co-schedule
jobs based on their load-balancing requirements and inter-processor
communication patterns.  In our earlier work~\cite{syp}, we had shown
that even when the pack-size is limited to $2$, co-scheduling based on
speed-up profiles can lead to faster workload completion and corresponding
savings in system energy.

Several recent
publications~\cite{sally,chandra2005predicting,onchipcosched} consider
co-scheduling at a single multicore node, when contention for resources
by co-schedu\-led tasks leads to complex tradeoffs between energy and
performance measures.  Chandra et al. \cite{chandra2005predicting}
predict and utilize inter-thread cache contention at a multicore in
order to improve performance. Hankendi and Coskun~\cite{onchipcosched}
show that there can be measurable  gains in energy per unit of work
through the application of their multi-level co-scheduling technique
at runtime which is based on
classifying tasks according to specific performance measures.
Bhaduria and McKee~\cite{sally} consider local search heuristics to
co-schedule tasks in a resource-aware manner at a multicore node to
achieve significant gains in thread throughput per watt.

These publications demonstrate that complex tradeoffs cannot be captured 
through the use of the speed-up measure alone, without significant additional measurements
to capture performance variations from cross-application interference at a
multicore node.  Additionally, as shown in our earlier work~\cite{syp}, we
expect significant benefits even when we aggregate only across multicore
nodes because speed-ups suffer due to 
of the longer
latencies of data transfer across nodes.  
We can therefore project savings in energy as being commensurate
with the savings in the time to complete a workload through co-scheduling.
Hence, we only test  configurations  where no more than a single application
can be scheduled on a multicore node. 


\section{Problem definition}
\label{sec.pb}

The application consists of  $n$ independent tasks $T_1,\dots,T_n$. The target 
execution platform consists of $p$ identical processors, and each task $T_i$ can be assigned 
an arbitrary number $\sigma(i)$ of 
processors, where $1 \leq \sigma(i) \leq p$. 
The objective is to minimize the total execution time by co-scheduling several tasks onto the $p$
resources. Note that the approach is agnostic of the granularity of each processor, which can be either
a single CPU or a multicore node.

\textbf{Speedup profiles --}
Let $t_{i,j}$ be the execution time of task $T_{i}$ with $j$
processors, and $\work(i,j) = j \times t_{i,j}$ be the corresponding work.
We assume the following for $1\leq i \leq n$ and $1\leq j < p$:\
\begin{equation}
\text{Non-increasing execution time:~} t_{i,j+1} \leq t_{i,j}
\label{hyp.monotony}
\end{equation}
\begin{equation}
\text{Non-decreasing work:~} \work(i,j) \leq \work(i,j+1)
\label{hyp.total.work}
\end{equation}
Equation~\eqref{hyp.monotony} implies that execution time is a 
non-increasing function of the number of processors. Equation~\eqref{hyp.total.work}
states that efficiency decreases with the number of enrolled processors: in other words, 
parallelization has a cost!  As a side note, we observe that these requirements make good sense
in practice: many scientific tasks $T_{i}$ are such that $t_{i,j}$ first decreases (due to load-balancing) and then increases (due to communication overhead), 
reaching a minimum for $j=j_{0}$; we can always let $t_{i,j} = t_{i,j_{0}}$ for $j \geq j_{0}$ by
never actually using more than $j_{0}$ processors for $T_{i}$.

\textbf{Co-schedules --}
A co-schedule partitions the $n$ tasks into groups (called 
\emph{\packs}), so that (i) all tasks from a given pack start their execution at the same time;
and (ii) two tasks from different \packs have disjoint execution intervals. See 
Figure~\ref{fig.coschedule} for an example.
The execution time, or \emph{cost}, of a \pack is the maximal execution time of a task in that \pack, and
the cost of a co-schedule is the sum of the costs of each \pack.

\textbf{\kinp optimization problem --} Given a fixed constant $k\leq p$, find a co-schedule with at most $k$ tasks per \pack that minimizes the execution time. The most general problem is when $k=p$,
but in some frameworks we may have an upper bound $k<p$ on the maximum number of tasks within
each \pack.

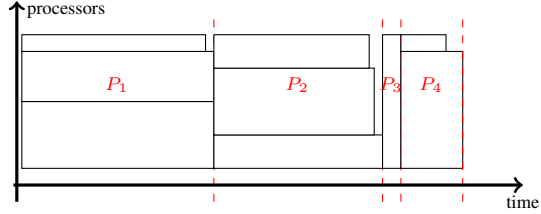
\begin{figure}[h]
\centering
\begin{tikzpicture}
\begin{scope}[scale=2/3]
\begin{scope}[xscale = 1/30, yscale=1/3]

\begin{scope}[shift={(0,0)}]
\draw[fill=white, draw = black] (0,0) rectangle (115,4) ; 
\draw[fill=white, draw = black] (0,4) rectangle (115,7) ; 
\draw[fill=white, draw = black] (0,7) rectangle (110,8) ; 
\draw[red] (57,5) node {$\scriptstyle{P_1}$};
\draw[red, dashed] (115,-2) --++(0,11) node {};
\end{scope}

\begin{scope}[shift={(115,0)}]
\draw[fill=white, draw = black] (0,0) rectangle (101,2) ; 
\draw[fill=white, draw = black] (0,2) rectangle (96,6) ; 
\draw[fill=white, draw = black] (0,6) rectangle (93,8) ; 
\draw[red] (50,5) node {$\scriptstyle{P_2}$};
\draw[red, dashed] (101,-2) --++(0,11) node {};
\end{scope}

\begin{scope}[shift={(216,0)}]
\draw[fill=white, draw = black] (0,0) rectangle (11,8) ;
\draw[red] (5,5) node {$\scriptstyle{P_3}$};
\draw[red, dashed] (11,-2) --++(0,11) node {};
\end{scope}

\begin{scope}[shift={(227,0)}]
\draw[fill=white, draw = black] (0,0) rectangle (37,7) ;
\draw[fill=white, draw = black] (0,7) rectangle (27,8) ; 
\draw[red] (18,5) node {$\scriptstyle{P_4}$};
\draw[red, dashed] (37,-2) --++(0,11) node {};
\end{scope}

\draw[->, very thick] (-5,-1) -- (300,-1) node [pos=1, below] {$\scriptstyle{\text{time}}$};
\draw[->, very thick] (-3,-2) -- (-3,10) node [pos=0.95, right] {$\scriptstyle{\text{processors}}$};
\end{scope}
\end{scope}
\end{tikzpicture}
\caption{A co-schedule with four packs $P_{1}$ to $P_{4}$.}
\label{fig.coschedule}
\end{figure}


\section{Theoretical results}
\label{sec.theory}

First we discuss the complexity of the problem in Section~\ref{sec.complexity}, 
by exhibiting polynomial and NP-complete instances. 
Next we discuss how to optimally schedule a set of $k$ tasks in a single \pack
(Section~\ref{sec.pack}). Then we explain how to compute the optimal solution (in expected exponential cost) in Section~\ref{sec.ilp}. 
Finally, we provide an approximation algorithm in Section~\ref{sec.approx}.

\subsection{Complexity}
\label{sec.complexity}

\begin{theorem}
The \kinp[1] and \kinp[2] problems can both be solved in polynomial time.
\end{theorem}

\begin{proof}
This result is obvious for \kinp[1]: each task is assigned exactly $p$ processors 
(see Equation~\eqref{hyp.monotony}) and the 
minimum execution time is $\sum_{i=1}^n t_{i,p}$.

This proof is more involved for \kinp[2], and we start with the \pinp[2] problem to get an intuition.
Consider the weighted undirected graph $G=(V,E)$, where $|V|=n$, each vertex  
$v_i \in V$ corresponding to a task $T_i$. The edge set $E$  is the following: (i)
for all $i$, there is a loop on $v_i$ of weight $t_{i,2}$; (ii) for all $i<i'$, there is an edge between $v_i$ and $v_{i'}$ of weight $\max (t_{i,1},t_{i',1})$.
Finding a perfect matching of minimal weight in  $G$ leads to the 
optimal solution to \pinp[2], which can thus be solved in polynomial time.

For the \kinp[2] problem, the proof is  similar, the only difference lies in the construction 
of the edge set  $E$: (i) for all $i$, there is a loop on $v_i$ of weight $t_{i,p}$; (ii) for all $i<i'$, there is 
an edge between $v_i$ and $v_{i'}$ of weight 
$\min_{j=1..p}\left (\max (t_{i,p-j},t_{i',j}) \right )$.
Again, a perfect matching of minimal weight in  $G$  gives 
the optimal solution to \kinp[2].  We conclude that the \kinp[2] problem 
can be solved in polynomial time.
\end{proof}

\begin{theorem}
	\label{thm.kinp}
When $k\geq 3$, the \kinp problem is strongly NP-complete.
\end{theorem}

\begin{proof}
We prove the NP-completeness of the decision problem associated to \kinp: given $n$ independent 
tasks, $p$ processors, a set of execution times $t_{i,j}$ for $1\leq i \leq n$ and $1\leq j \leq p$ 
satisfying Equations~\eqref{hyp.monotony} and~\eqref{hyp.total.work}, a fixed constant $k\leq p$ 
and a deadline~$D$, can we find a co-schedule with at most $k$ tasks per \pack, and whose 
execution time does not exceed~$D$?
The problem is obviously in NP: if we have the composition of every \pack, and for each task in a 
\pack, the number of processors onto which it is assigned, we can verify in 
polynomial time: (i) that it is indeed a \pack schedule; (ii) that the execution time is smaller than a 
given deadline.

We first prove the strong completeness of \kinp[3]. We use a reduction from \tpart. 
Consider an arbitrary instance $\II_1$ of \tpart: 
given an integer $B$ and $3n$ integers $a_1,\ldots,a_{3n}$, can we partition the $3n$ integers 
into $n$ triplets, each of sum $B$? We can assume that $\sum_{i=1}^{3n} a_{i} = nB$, 
otherwise  $\II_1$ has no solution.
The \tpart problem is NP-hard in the strong sense~\cite{GareyJohnson},
which implies that we can encode all integers ($a_1$, \ldots, $a_{3n}$, $B$)  in unary.
We build the following instance $\II_2$ of \kinp[3]: the number of processors is $p=B$, the 
deadline is $D=n$, there are $3n$ tasks $T_i$, with the following execution times: 
for all $i,j$, if $j<a_i$ 
then $t_{i,j} =1 + \frac{1}{a_i}$, otherwise $t_{i,j} = 1$.
It is easy to check that  Equations~\eqref{hyp.monotony} and~\eqref{hyp.total.work} are both satisfied.
For the latter, since there are only two possible execution times for each task,
we only need to check Equation~\eqref{hyp.total.work} for $j=a_{i}-1$, and we do obtain that 
$(a_{i}-1)(1 + \frac{1}{a_i})  \leq a_{i}$.  Finally, $\II_2$ has a size 
polynomial in the size of  $\II_1$, even if we write all instance parameters in unary:
the execution time is $n$, and the $t_{i,j}$ have the same size
as the~$a_i$.
 
We now prove that $\II_1$ has a solution if and only if $\II_2$ does.
Assume first that $\II_1$ has a solution. 
For each triplet $(a_i,a_j,a_k)$ of $\II_1$, we create a \pack with the three tasks $(T_i,T_j,T_k)$
where $T_i$ is scheduled on $a_i$ processors, $T_j$  on $a_j$ processors, and
$T_k$ on $a_k$ processors. By definition, we have $a_i+a_j+a_k = B$, and the 
execution time of this \pack is $1$. We do this for the $n$ triplets, which gives a valid
co-schedule whose total execution time~$n$. Hence the solution to $\II_2$.

Assume now that $\II_2$ has a solution. The minimum execution time for any \pack is $1$ (since it is the minimum execution time of any 
task and a \pack cannot be empty). Hence the solution cannot have more than $n$ \packs. 
Because there are $3n$ tasks and the number of elements in a \pack is limited to three, there are
exactly $n$ \packs, each of exactly $3$ elements, and furthermore all these \packs have an execution time of 
$1$ (otherwise the deadline $n$ is not matched). If there were a \pack $(T_i,T_j,T_k)$ such that
$a_i+a_j+a_k>B$, then one of the three tasks, say $T_{i}$, would have to use fewer than $a_{i}$ processors, hence would have
an execution time greater than $1$.
Therefore, for each \pack $(T_i,T_j,T_k)$, we have $a_i+a_j+a_k\leq B$. The fact that this inequality is an equality 
for all \packs follows from the fact that $\sum_{i=1}^{3n} a_{i} = nB$.
Finally, we conclude by saying that the set of triplets $(a_i,a_j,a_k)$ for every \pack $(T_i,T_j,T_k)$ is a solution to $\II_1$.

The final step is to prove the completeness of  \kinp[k] for a given $k \geq 4$.
We perform a similar reduction from the same instance $\II_{1}$ of \tpart.
We construct the instance $\II_2$ of \kinp where the number of processors is $p=B +(k-3)(B+1)$ and the 
deadline is~$D=n$. There are $3n$ tasks $T_i$ with the same execution times as before (for $1\leq i\leq 3n$, if $j<a_i$ then $t_{i,j} =1 + \frac{1}{a_i}$, otherwise $t_{i,j} = 1$), and also $n(k-3)$ new identical tasks such that, for $3n+1\leq i \leq kn$,  $t_{i,j}=\max \left( \frac{B+1}{j}, 1 \right )$. It is easy to check that  Equations~\eqref{hyp.monotony} and~\eqref{hyp.total.work} are also fulfilled for the new tasks.
If $\II_1$ has a solution, we construct the solution to $\II_2$ similarly to the previous reduction, 
and we add to each \pack $k-3$ tasks $T_i$ with $3n+1 \leq i \leq kn$, each assigned to $B+1$ processors. This solution has an execution time
exactly equal to $n$.
Conversely, if $\II_2$ has a solution, we can verify that there are exactly $n$ \packs (there are 
$kn$ tasks and each \pack has an execution time at least equal to~$1$). Then we can verify that there are
at most $(k-3)$ tasks~$T_i$ with $3n+1 \leq i \leq kn$ per \pack, since there are exactly $(k-3)(B+1) + B$ processors. Otherwise, if there were $k-2$ (or more) such tasks in a \pack, then one of them would be scheduled on less than $B+1$ processors, and the execution time of the \pack would be greater than~$1$. Finally, we can see that in~$\II_2$, each \pack is composed of $(k-3)$ tasks~$T_i$ with $3n+1 \leq i \leq kn$, scheduled
on $(k-3)(B+1)$ processors at least, and that there remains triplets of tasks~$T_i$, with $1\leq i \leq 3n$, scheduled on at most $B$ processors. The end of the proof is identical to the reduction in the case $k=3$. 
\end{proof}

Note that the \kinp[3] problem is NP-complete, and the \kinp[2] problem can be solved in polynomial time, hence \pinp[3] is the simplest problem whose complexity remains open.

\subsection{Scheduling a \pack of tasks}
\label{sec.pack}

In this section, we discuss how to optimally schedule a set of $k$~tasks in a single \pack: the $k$ tasks 
$T_1,\ldots,T_k$ are given, and we search for an assignment function
 $\sigma: \{1,\ldots,k\} \rightarrow \{1,\ldots,p\}$ such that $\sum_{i=1}^k \sigma(i) \leq p$, where 
 $\sigma(i)$ is the number of processors assigned to task~$T_i$. 
Such a schedule is called a \bs, and its \emph{cost} is $\max_{1\leq i\leq k} t_{i,\sigma(i)}$.
In Algorithm~\ref{algo.opt.pack} below, we use the notation $T_{i} \leqb T_{j}$ if 
$t_{i,\sigma(i)} \leq t_{j,\sigma(j)}$:

\begin{algorithm}[htb]
\caption{Finding the optimal \bs $\sigma$ of $k$ tasks in the same \pack.}
\label{algo.opt.pack}
procedure {Optimal-\bs}($T_1, \ldots,T_k$)\\
\Begin{
\For{$i=1$ to $k$}{
$\sigma(i) \leftarrow 1$
}
Let $L$ be the list of tasks sorted in non-increasing values of \leqb\;
$p_{\text{available}} := p-k$\;
\While{$p_{\text{available}} \neq 0$}{
$T_{i^{\star}}:=\text{head}(L)$\;
$L:=\text{tail}(L)$\;
$\sigma(i^{\star}) \leftarrow \sigma(i^{\star}) +1$\;
$p_{\text{available}} := p_{\text{available}} -1$\;
$L:=$ Insert $T_{i^{\star}}$ in $L$ according to its \leqb value\;
}
\Return{$\sigma$;}
}
\end{algorithm}

\begin{theorem}
\label{th.1bs}
Given $k$ tasks to be scheduled on $p$ processors in a single \pack, 
Algorithm~\ref{algo.opt.pack} finds a \bs of minimum cost in time~$O(p \log(k))$.
\end{theorem}

In this greedy algorithm, we first assign one processor to each task, and while there are processors 
that are not processing any task, we  select the task with the longest execution time and assign an 
extra processor to this task. 
Algorithm~\ref{algo.opt.pack} performs $p-k$ iterations to assign the extra processors. We denote by
$\si{\ell}$ the current value of the function $\sigma$ at the end of iteration~$\ell$. 
For convenience, we let $t_{i,0}=+\infty$ for $1\leq i \leq k$. We start with the following lemma:

\emph{Lemma}: At the end of iteration $\ell$ of Algorithm~\ref{algo.opt.pack}, let 
$T_{i^{\star}}$ be the first task of the sorted list, i.e., the task with longest execution time. Then, for all~$i$,  
$t_{i^{\star},\si{\ell}(i^{\star})} \leq t_{i,\si{\ell}(i)-1}$. 
\begin{proof}
Let $T_{i^{\star}}$ be the task with longest execution time at the end of iteration~$\ell$. 
For tasks such that $\si{\ell}(i)=1$, the result is obvious since $t_{i,0}=+\infty$. 
Let us consider any task~$T_i$ such that $\si{\ell}(i)>1$. 
Let $\ell'+1$ be the last iteration when a new processor was assigned to task~$T_i$:  
$\si{\ell'}(i)=\si{\ell}(i)-1$ and $\ell'<\ell$. By definition of iteration $\ell'+1$, task~$T_i$ was chosen because 
$t_{i,\si{\ell'}(i)}$ was greater than any other task, in particular 
$t_{i,\si{\ell'}(i)} \geq t_{i^{\star},\si{\ell'}(i^{\star})}$.
Also, since we never remove processors from tasks, we have $\si{\ell'}(i) \leq \si{\ell}(i)$ and $\si{\ell'}(i^{\star}) \leq \si{\ell}(i^{\star})$. 
Finally, $t_{i^\star,\si{\ell}(i^{\star})}\leq t_{i^\star,\si{\ell'}(i^{\star})} \leq t_{i,\si{\ell'}(i)} =  t_{i,\si{\ell}(i)-1}$.
\end{proof}
We are now ready to prove Theorem~\ref{th.1bs}.

\begin{proof}[of Theorem~\ref{th.1bs}]
Let $\sigma$ be the \bs returned by Algorithm~\ref{algo.opt.pack} of cost $c(\sigma)$, 
and let $T_{i^{\star}}$ be a task such that  $c(\sigma)=t_{i^\star,\sigma(i^\star)}$.
Let $\sigma'$ be a \bs of cost $c(\sigma')$. We prove below that  $c(\sigma') \geq c(\sigma)$, 
hence $\sigma$ is a \bs of minimum cost:
\begin{compactitem}
	\item If $\sigma'(i^{\star}) \leq \sigma(i^{\star})$, then $T_{i^{\star}}$  has fewer processors 
	in $\sigma'$ than in $\sigma$, hence its execution time is larger, and $c(\sigma') \geq c(\sigma)$. 
	\item If $\sigma'(i^{\star}) > \sigma(i^{\star})$, then there exists~$i$ such that 
$\sigma'(i) < \sigma(i)$ (since the total number of processors is~$p$ in both $\sigma$ and $\sigma'$). We can apply the previous  Lemma at the end of the last iteration, where $T_{i^{\star}}$ is the task of maximum execution time: 
$t_{i^{\star},\sigma(i^{\star})} \leq t_{i,\sigma(i)-1} \leq t_{i,\sigma'(i)}$, and therefore $c(\sigma') \geq c(\sigma)$.
\end{compactitem}
Finally, the time complexity is obtained as follows: first we sort $k$ elements, in time $O(k \log k)$. Then there are $p-k$ iterations, and at each iteration, we insert an element in a sorted list of $k-1$ elements, which takes $O(\log k)$ operations (use a heap for the data structure of $L$). 
\end{proof}

Note that it is easy to compute an optimal \bs using a dynamic-programming algorithm: the optimal cost is 
$c(k,p)$, which we compute using the recurrence formula
$$c(i,q)=\min_{1\leq q' \leq q} \{ \max(c(i-1,q-q'),t_{i,q'}) \}$$ for $2\leq i\leq k$ 
and $1\leq q\leq p$, initialized by $c(1,q)=t_{1,q}$, and $c(i,0)=+\infty$. 
The complexity of this algorithm is $O(kp^2)$. However, we can significantly reduce the complexity of this algorithm by using Algorithm~\ref{algo.opt.pack}.

\subsection{Computing the optimal solution}
\label{sec.ilp}

In this section we sketch two methods to find the optimal solution to the general \kinp problem.
This can be useful to solve some small-size instances, albeit at the price of a cost exponential 
in the number of tasks $n$.

The first method is to generate all possible partitions of the tasks into packs. This amounts to computing
all partitions of $n$ elements into subsets of cardinal at most $k$. For a given partition 
of tasks into packs, we use Algorithm~\ref{algo.opt.pack} to find the optimal processor
assignment for each pack, and we can compute the optimal cost for the partition. There remains to take the
minimum of these costs among all partitions.

The second method is to cast the problem in terms of an integer linear program:
\begin{theorem}
The following integer linear program characterizes the \kinp problem, where the unknown 
variables are the $x_{i,j,b}$'s (Boolean variables) and the $y_b$'s (rational variables), for $1 \leq i,b \leq n$ and  $1 \leq j \leq p$: 
\begin{eqnarray}
\label{lin.prog.kinp}
\begin{array}{ll}
\text{Minimize} \sum_{b=1}^{n}  y_b & \text{subject to}\\
\text{(i) } \sum_{j,b} x_{i,j,b} = 1, & 1\leq i \leq n\\
\text{(ii) } \sum_{i,j} x_{i,j,b} \leq k, & 1\leq b \leq n\\
\text{(iii) } \sum_{i,j} j \times x_{i,j,b} \leq p, & 1\leq b \leq n\\
\text{(iv) } x_{i,j,b}\times t_{i,j} \leq y_b, & 1\leq i,b \leq n, 1 \leq j \leq p\\
\end{array}
\end{eqnarray}
\end{theorem}

\begin{proof}
The $x_{i,j,b}$'s are such that $x_{i,j,b}=1$ if and only if task~$T_i$ is in the \pack~$b$ and it is executed on $j$~processors; $y_b$ is the execution time of \pack~$b$. Since there are no more than $n$ \packs (one task per \pack), $b\leq n$. The sum $\sum_{b=1}^{n}  y_b$ is therefore the total execution time ($y_b=0$ if there are no tasks in \pack~$b$). 
Constraint~(i) states that each task is assigned to exactly one \pack~$b$, and with one number of processors~$j$. Constraint~(ii) ensures that there are not more than $k$~tasks in a \pack. Constraint~(iii) adds up the number of processors in \pack~$b$, which should not exceed~$p$. Finally, constraint~(iv) computes the cost of each \pack.  
\end{proof}

\subsection{Approximation algorithm}
\label{sec.approx}
\label{sec.makepack}

In this section we introduce \ok, a $3$-approximation algorithm for the \pinp problem. 
The design principle of \ok is the following: we start from the assignment 
where each task is executed on one processor, and use Algorithm~\ref{algo.makepack} to build a 
first solution.  Algorithm~\ref{algo.makepack} is a greedy heuristic that builds a co-schedule
when each task is pre-assigned a number of processors for execution. 
Then we iteratively refine the solution, adding a processor to the task with longest execution time,
and re-executing Algorithm~\ref{algo.makepack}. Here are details on both algorithms:

\emph{Algorithm~\ref{algo.makepack}.} The   \kinp problem with processor pre-assignments remains strongly NP-complete (use a similar reduction as
in the proof of Theorem~\ref{thm.kinp}). We propose a greedy procedure
in Algorithm~\ref{algo.makepack} which is similar to the First Fit Decreasing Height algorithm
for strip packing~\cite{coffman1980performance}. The output is a co-schedule
with at most $k$ tasks per pack, and the complexity is $O(n \log(n))$ (dominated by sorting).

\emph{Algorithm~\ref{algo.optimal.known}.} 
We iterate the calls to Algorithm~\ref{algo.makepack}, adding a processor to the task with longest 
execution time, until: (i) either the task of longest execution time is already assigned $p$ processors, 
or (ii) the sum of the work of all tasks is greater than $p$ times the longest execution time. The 
algorithm returns the minimum cost found during execution. The complexity of this algorithm is 
$O(n^2p)$ (in the calls to Algorithm~\ref{algo.makepack} we do not need to re-sort the list but 
maintain it sorted instead) in the simplest version presented here, but can be reduced to 
$O(n\log(n) + np)$ using standard algorithmic techniques.

\begin{algorithm}[htb]
\caption{Creating \packs of size at most $k$, when the number $\sigma(i)$ of processors per task $T_{i}$ is fixed.}
\label{algo.makepack}
procedure {\makebatch}($n,p,k,\sigma$)\\
\Begin{
Let $L$ be the list of tasks sorted in non-increasing values of execution times $t_{i,\sigma(i)}$\;
\While{$L \neq \emptyset$}{
  Schedule the current task on the first \pack with enough available processors and fewer than $k$ tasks. 
  Create a new \pack if no existing \pack fits\;
  Remove the current task from~$L$\;	
}
\Return{the set of \packs}
}
\end{algorithm}

\begin{algorithm}[htb]
\caption{\ok\label{algo.ok}}
\label{algo.optimal.known}
procedure {\ok}($T_1, \ldots,T_n$)\\
\Begin{

$\MS = +\infty$ \;

\lFor{$j=1$ to $n$}{
$\sigma(j) \leftarrow 1$\; }

\For{$i=0$ to $n(p-1)-1$}{
Let $A_{\text{tot}}(i) = \sum_{j=1}^n t_{j,\sigma(j)} \sigma(j)$\;
Let $T_{j^{\star}}$ be one task that maximizes $t_{j,\sigma(j)}$\; 
Call \makebatch($n,p,p,\sigma$)\; 

Let $\MS_i$ be the cost of the co-schedule\;
\lIf{$\MS_i< \MS
$}{$\MS \leftarrow \MS_i$\;}
\lIf{  ($\frac{A_{\text{tot}}(i)}{p}>t_{j^{\star},\sigma(j^{\star})}$ ) or ($\sigma(j^{\star})=p$)}{\Return{$\MS$;} \tcc{Exit loop}}
\lElse{$\sigma(j^{\star})\leftarrow \sigma(j^{\star})+1$;   \tcc{Add a processor to $T_{j^{\star}}$}}
}

\Return{$\MS$;}
}
\end{algorithm}

\begin{theorem}
\ok is a 3-approximation algorithm for the \pinp problem.
\end{theorem}

\begin{proof}
We start with some notations: 
\begin{compactitem}
	\item step $i$ denotes the $i^{th}$ iteration of the main loop of Algorithm~\ok;
	\item $\si{i}$ is the allocation function at step $i$;
	\item $t_{\max}(i) = \max_{j} t_{j,\si{i}(j)}$ is the maximum execution time of any task at step $i$;
	\item $j^{\star}(i) $ is the index of the  task with longest execution time at step $i$ (break ties arbitrarily);
	\item $A_{\text{tot}}(i) =\sum_j t_{j,\si{i}(j)} \si{i}(j)$ is the total work that has to be done at step $i$;
	\item $\MS_i$ is the result of the scheduling procedure at the end of step $i$;
	\item $\opt$ denotes an optimal solution, with allocation function $\si{\opt}$, 
	execution time~$\MS_\opt$, and total work 
$$A_{\opt} = \sum_j t_{j,\si{\opt}(j)} \si{\opt}(j).$$ 
\end{compactitem}

\medskip
\noindent Note that there are three different ways to exit algorithm~\ok:
\begin{compactenum}
	\item If we cannot add processors to the task with longest execution time, i.e.,  
$\si{i}(j^{\star}(i))=p$;
	\item If $\frac{A_{\text{tot}}(i)}{p}>t_{\max}(i)$ after having computed the execution time for 
this assignment;
	\item When each task has been assigned $p$ processors (the last step of the loop ``for": we 
have assigned exactly $np$ processors, and no task can be assigned more than $p$ processors).
\end{compactenum}

\begin{lemma}
	\label{lem.3}
At the end of step $i$, $\MS_i \!\leq\! 3 \! \max \!\left( \! t_{\max}(i),\!\frac{A_{\text{tot}}(i)}{p} \!\right)$.
\end{lemma}

\begin{proof}
Consider the \packs returned by Algorithm~\ref{algo.makepack}, 
sorted by non-increasing execution times, $B_1,B_2,\ldots, B_n$ (some of the \packs
may be empty, with an execution time~$0$). 
Let us denote, for $1\leq q\leq n$,
\begin{compactitem}
	\item $j_q$ the task with the longest execution time of \pack $B_q$ (i.e., the first task scheduled on 
$B_q$);
	\item $t_q$ the execution time of \pack $B_q$ (in particular, $t_q \!=\! t_{j_q,\si{i}(j_q)}$\!);
	\item $A_q$ the sum of the task works in \pack $B_q$;
	\item $p_q$ the number of processors available in \pack $B_q$ when $j_{q+1}$ was 
scheduled in \pack $B_{q+1}$.
\end{compactitem}

With these notations, $\MS_i = \sum_{q=1}^n t_q$ and $A_{\text{tot}}(i) = \sum_{q=1}^n \!A_q$.
For each \pack, note that $p t_q \geq A_q$, since $p t_q$ is the maximum work that can be done
on $p$ processors with an execution time of~$t_q$.
Hence, $\MS_i \geq \frac{A_{\text{tot}}(i)}{p}$.

In order to bound $\MS_i$, let us first remark that $\si{i}(j_{q+1}) > p_q$:
otherwise $j_{q+1}$ would have been scheduled on \pack $B_q$.
Then, we can exhibit a lower bound for $A_q$, namely $A_q \geq t_{q+1} (p-p_q)$. Indeed, the 
tasks scheduled before $j_{q+1}$ all have a length greater than $t_{q+1}$ by definition.
Furthermore, obviously $A_{q+1} \geq t_{q+1} p_q$ (the work of the first task scheduled in \pack 
$B_{q+1}$). So finally we have, $A_q + A_{q+1} \geq  t_{q+1} p$.

Summing over all $q$'s, we have: $2\sum_{q=1}^n \frac{A_q}{p} \geq \sum_{q=2}^n t_q$, hence
$2\frac{A_{\text{tot}}(i)}{p} +t_1 \geq \MS_i$. Finally, note that $t_1 = t_{\max}(i)$, and therefore
$\MS_i \leq 3 \max \left (t_{\max}(i),\frac{A_{\text{tot}}(i)}{p} \right )$.
Note that this proof is similar to the one for the Strip-Packing problem 
in~\cite{coffman1980performance}.
\end{proof}

\begin{lemma}
	\label{lem.monot}
At each step~$i$, $A_{\text{tot}}(i+1) \geq A_{\text{tot}}(i)$ and $t_{\max}(i+1)\leq t_{\max}(i)$, i.e., 
the total work is increasing and the maximum execution time is decreasing. 
\end{lemma}

\begin{proof}
$A_{\text{tot}}(i+1) = A_{\text{tot}}(i)- a+b$, where
\begin{compactitem}
\item $a=\work(j^{\star}(i),\si{i}(j^{\star}(i))) $, and
\item $b=\work(j^{\star}(i),\si{i+1}(j^{\star}(i)))$.
\end{compactitem}
But $b =\work(j^{\star}(i),\si{i}(j^{\star}(i))+1)$
and $a \leq b$ by Equation~\eqref{hyp.total.work}.
Therefore, $A_{\text{tot}}(i+1) \geq A_{\text{tot}}(i)$.
Finally, $t_{\max}(i+1)\leq t_{\max}(i)$ since only one of the tasks with the longest 
execution time is modified, and its execution time can only decrease thanks 
to Equation~\eqref{hyp.monotony}.
\end{proof}


\begin{lemma}
	\label{prop.opt}
Given an optimal solution \opt, 
$\forall j, t_{j,\si{\opt}(j)} \leq \MS_\opt$ and 
$A_{\opt} \leq p\MS_\opt$.
\end{lemma}

\begin{proof}
The first inequality is obvious. As for the second one, $p\MS_\opt$ is the maximum 
work that can be done on $p$ processors within an execution time of $\MS_\opt$, hence it must not be smaller than $A_{\opt}$, which is the sum of the work of the tasks with the optimal allocation.
\end{proof}

\begin{lemma}
	\label{lem.strict}
For any step~$i$ such that $t_{\max}(i) > \MS_{\opt}$, then $\forall j, \si{i}(j) \leq \si{\opt}(j)$, 
and $A_{\text{tot}}(i) \leq A_{\opt}$.
\end{lemma}
\begin{proof}
Consider a task $T_j$.
If $\si{i}(j) =1$, then clearly $\si{i}(j) \leq \si{\opt}(i)$.
Otherwise, $\si{i}(j) >1$, and then by definition of the algorithm, there was a step $i'<i$, such 
that $\si{i'}(j) = \si{i}(j) -1$ and $\si{i'+1}(j) = \si{i}(j)$. 
Therefore $t_{\max}(i') = t_{j,\si{i'}(j)}$.
Following Lemma~\ref{lem.monot}, we have $t_{\max}(i') \geq t_{\max}(i) > \MS_{\opt}$. 
Then necessarily, $\si{\opt}(j)>\si{i'}(j)$, hence the result.
Finally, $A_{\text{tot}}(i) \leq A_{\opt}$ is a simple corollary of the previous result and of 
Equation~\eqref{hyp.total.work}.
\end{proof}

\begin{lemma}
	\label{lem.keepon}
For any step $i$ such that $t_{\max}(i) > \MS_{\opt}$, then $\frac{A_{\text{tot}}(i)}{p} <t_{\max}(i)$.
\end{lemma}
\begin{proof}
Thanks to Lemma~\ref{lem.strict}, we have $\frac{A_{\text{tot}}(i)}{p} \leq \frac{A_{\opt}}{p}$. 
Lemma~\ref{prop.opt} gives us $\frac{A_{\opt}}{p} \leq \MS_{\opt}$, 
hence the result.
\end{proof}

\begin{lemma}
	\label{lem.exists_i0}
There exists $i_0 \geq 0$ such that $t_{\max}(i_0-1) > \MS_{\opt}\geq t_{\max}(i_0)$ 
(we let $t_{\max}(-1) = +\infty$).
\end{lemma}

\begin{proof}
We show this result by contradiction.
Suppose such $i_0$ does not exist. Then $t_{\max}(0) > \MS_{\opt}$ (otherwise $i_0=0$ would
suffice). Let us call $i_1$ the last step of the run of the algorithm. Then by induction we have the 
following property, $t_{\max}(0) \geq t_{\max}(1) \geq \cdots \geq t_{\max}(i_1) > \MS_{\opt}$
(otherwise $i_0$ would exist, hence contradicting our hypothesis).
Recall that there are three ways to exit the algorithm, hence three possible definitions for~$i_1$: 
\begin{itemize}
	\item $\si{i_1}(j^{\star}(i_1))= p$, however then we would have 
$t_{\max}(i_1)  = t_{j^{\star}(i_1),p}> \MS_{\opt} \geq t_{j^{\star}(i_1),\si{\opt}}$ 
(according to Lemma~\ref{prop.opt}). 
This contradicts Equation~\eqref{hyp.monotony}, which states that 
$t_{j^{\star}(i_1),p} \leq t_{j^{\star}(i_1),k}$ for all $k$.
	\item $i_1 = n(p-1) -1$, but then we have the same result, i.e., $\si{i_1}(j^{\star}(i_1))= p$ 
because this is true for all tasks.
	\item $t_{\max}(i_1) < \frac{A_{\text{tot}}(i_1)}{p}$, but this is false according to 
Lemma~\ref{lem.keepon}.
\end{itemize}
We have seen that \ok could not have terminated at step~$i_1$, however since \ok terminates 
(in at most $n(p-1)-1$ steps), we have a contradiction. Hence we have shown the existence of~$i_0$.
\end{proof}

\begin{lemma}
	\label{lem.atot_i0}
$A_{\text{tot}}(i_0) \leq A_{\opt}$.
\end{lemma}

\begin{proof}
Consider step $i_0$. 
If $i_0=0$, then at this step, all tasks are scheduled on exactly one 
processor, and $\forall j, \si{i_0}(j) \leq \si{\opt}(j)$. Therefore,  
$A_{\text{tot}}(i_0) \leq A_{\opt}$.
If $i_0 \neq 0$, consider  step $i_0-1$: $t_{\max}(i_0-1) > \MS_{\opt}$. 
From Lemma~\ref{lem.strict}, we have $\forall j, \si{i_0-1}(j) \leq \si{\opt}(j)$. 
Furthermore, it is easy to see that 
$\forall j \neq j^{\star}(i_0-1), \si{i_0}(j)=\si{i_0-1}(j)$ since no task other than 
$j^{\star}(i_0-1)$ is modified. We also have the following properties: 
\begin{itemize}
	\item $t_{j^{\star}(i_0-1),\si{i_0-1}( j^{\star}(i_0-1))}=t_{\max}(i_0-1)$;
	\item $t_{\max}(i_0-1) > t_{\opt}$ (by definition of step $i_0$);
	\item $t_{\opt} \geq t_{j^{\star}(i_0-1),\si{\opt}( j^{\star}(i_0-1))}$ (Lemma~\ref{prop.opt});
	\item $\si{i_0}( j^{\star}(i_0-1)) = \si{i_0-1}( j^{\star}(i_0-1)) +1$.
\end{itemize}
The three first properties and Equation~\eqref{hyp.monotony} allow us to say that $\si{i_0-1}( j^{\star}(i_0-1)) < \si{\opt}( j^{\star}(i_0-1))$. Thanks to the fourth property, 
$\si{i_0}( j^{\star}(i_0-1)) \leq \si{\opt}(j)$.
Finally, we have, for all $j, \si{i_0}(j)\leq \si{\opt}(j)$, and therefore 
$A_{\text{tot}}(i_0) < A_{\opt}$ by Equation~\eqref{hyp.total.work}.
\end{proof}

We are now ready to prove the theorem. For $i_0$ introduced in Lemma~\ref{lem.exists_i0}, we have:
\begin{align*}
\MS_{i_0} &\leq 3 \max \left ( t_{\max}(i_0),\frac{A_{\text{tot}}(i_0)}{p}\right ) \\
&\leq 3 \max \left ( \MS_{\opt},\frac{A_{\opt}}{p} \right)\\
& \leq 3 \MS_{\opt}
\end{align*}
The first inequality comes from Lemma~\ref{lem.3}. The second inequality is due to 
Lemma~\ref{lem.exists_i0} and~\ref{lem.atot_i0}. The last inequality comes from 
Lemma~\ref{prop.opt}, hence the final result. 
\end{proof}

\section{Heuristics}
\label{sec.heuristics}

In this section, we describe the heuristics that we use to solve the \kinp problem.

\noindent
\textbf{\randompack--} In this heuristic,  we generate the \packs randomly: as long as there remain tasks, randomly choose an integer $j$ between 1 and $k$, and then  randomly select $j$ tasks to form a \pack. Once the \packs are generated, apply Algorithm~\ref{algo.opt.pack} to optimally schedule each of them.

\noindent
\textbf{\randomproc--}  In this heuristic, we assign the number of processors to 
each task randomly between $1$ and~$p$, then use Algorithm~\ref{algo.makepack}   
to generate the \packs, followed by Algorithm~\ref{algo.opt.pack} on each \pack. 

\noindent
\textbf{A word of caution--} We point out that  \randompack and \randomproc are not pure random
heuristics, in that they already benefit from the theoretical results of Section~\ref{sec.theory}. A more 
naive  heuristic would pick both a task and a number of processor randomly, and greedily build packs,
creating a new one as soon as more than $p$ resources are assigned within the current pack. 
Here, both \randompack and  \randomproc use the optimal resource allocation strategy (Algorithm~\ref{algo.opt.pack}) within a pack; in addition,  \randomproc
uses an efficient partitioning algorithm (Algorithm~\ref{algo.makepack}) to create packs when resources
are pre-assigned to tasks.

\noindent
\textbf{\ok--} This heuristic is an extension of Algorithm~\ref{algo.ok} in Section~\ref{sec.makepack}
to deal with \packs of size $k$ rather than $p$: simply call \makebatch($n,p,k,\sigma$) instead of
\makebatch($n,p,p,\sigma$). However, although
we keep the same name as in Section~\ref{sec.makepack} for simplicity, we point out that
it is unknown whether
this heuristic is a $3$-approximation algorithm for arbitrary $k$.

\noindent
\textbf{\bbb($\varepsilon$)--} The rationale for this heuristic  
is to create \packs that are well-balanced: the difference between the smallest and  longest execution times in each \pack should be as small
as possible. Initially, we assign one processor per task (for $1\leq i \leq n$, $\sigma(i)=1$), and tasks are sorted into a list $L$ ordered by non-increasing execution times (\leqb values). 
While there remain some tasks in $L$, let $T_{i^\star}$ be the first task of the list, and let $t_{\max} = t_{i^\star,\sigma(i^\star)}$. Let $V_{req}$ be 
the ordered set of tasks $T_{i}$ such that $t_{i,\sigma(i)} \geq (1-\varepsilon)t_{\max}$:
this is the sublist of tasks (including $T_{i^\star}$ as its first element)
whose execution times are close
to the longest execution time $t_{\max}$, and $\varepsilon \in [0,1]$ is some parameter.
Let $p_{req}$ be the total number of processors requested by tasks in $V_{req}$. 
%
If $p_{req} \geq p$, a new \pack is created greedily with  the first tasks of $V_{red}$, 
adding them into the \pack while there are no more than $p$ processors used 
and no more than $k$~tasks in the \pack.
The corresponding tasks are removed from the list $L$. 
Note that $T_{i^\star}$ is always inserted in the created \pack. 
Also, if we have $\sigma(i^{\star})=p$, then a new \pack with only $T_{i^{\star}}$ is created.
Otherwise ($p_{req} < p$), 
an additional processor is assigned to the (currently) critical task $T_{i^\star}$, hence
$\sigma(i^{\star}):= \sigma(i^{\star})+1$, and the process iterates
after the list~$L$ is updated with the insertion of the new value for $T_{i^{\star}}$
Finally, once all \packs are created, we apply Algorithm~\ref{algo.opt.pack} in each \pack, so as to derive the optimal schedule within each \pack. 


We have $0<\varepsilon<1$. A small value of $\varepsilon$ will lead to balanced \packs, but may end up with a single task with $p$ processors per \pack. Conversely,  a large value of $\varepsilon$ will create new \packs more easily, i.e., with fewer processors per task. 
The idea is therefore to call the heuristic with different values of $\varepsilon$, and to select the solution that leads to the best execution time.


\noindent
\textbf{Summary of heuristics--} 
We consider two variants of the random heuristics, either with one single run, 
or with $9$ different runs, hence hoping to obtain a better solution, at the
price of a slightly longer execution time. These heuristics are denoted respectively
\randompack-1, \randompack-9, \randomproc-1, \randomproc-9. 
Similarly, for \bbb, we either use one single run with $\varepsilon=0.5$ (\bbb-1),
or $9$ runs with $\varepsilon \in \{.1, .2, \ldots, .9\}$ (\bbb-9). Of course,
there is only one variant of \ok, hence leading to seven heuristics.

\noindent
\textbf{Variants--} 
We have investigated variants of \bbb, trying to make a better choice than the
greedy choice to create the packs, for instance using a dynamic programming
algorithm to minimize processor idle times in the pack. However, there was
very little improvement at the price of a much higher running time of the heuristics. 
Additionally, we tried to improve heuristics with up to $99$ runs, both for the random ones
and for \bbb, but here again, the gain in performance was negligible compared 
to the increase in running time. Therefore we present 
only results for these seven heuristics in the following.

%
%

%

\section{Experimental Results}
\label{sec:exp}

In this section, we study the performance of the seven heuristics 
on workloads of parallel tasks. First we describe the workloads,
whose application execution times model profiles of parallel scientific codes.
Then we present the measures used to evaluate the quality of the schedules,
and finally we discuss the results. 


\noindent
\textbf{Workloads--} 
Workload-I
corresponds to $10$ parallel scientific applications  that involve
VASP~\cite{vasp}, ABAQUS~\cite{abaqus}, LAMMPS~\cite{lammps} and
Petsc~\cite{petsc}.  The execution times of these applications were
observed on a cluster with Intel Nehalem 8-core nodes connected by
a QDR Infiniband network with a total of $128$ cores. In other words,
we have  $p=16$ processors, and each processor is a multicore node.\\
Workload-II   is a {\em synthetic} test suite that was designed to
represent a larger set of scientific applications. It models tasks
whose parallel execution time for a fixed problem size $m$ on $q$
cores is of the form $t(m,q) =  f \times t(m,1) + (1-f) \frac{t(m,1)}{q}
+ \kappa(m,q)$, where $f$ can be interpreted as the inherently serial
fraction, and $\kappa$ represents overheads related to synchronization
and the communication of data. We consider tasks with sequential times
$t(m,1)$ of the form $cm$, $cm\log_2 n$, $cm^2$ and $cm^3$, where $c$ is
a suitable constant. We consider values of $f$ in $\{0, 0.04, 0.08,.16,
.32\}$, with overheads $\kappa(m,q)$ of the form $\log_2 q$, $(\log _2
q)^2$, $q \log_2 q$, $\frac{m}{q}\log_2 q$, $\sqrt{m/q}$, and $m \log_2 q$
to create a workload with $65$ tasks executing on up to $128$ cores.\\
The same process was also used to develop Workload-III, our largest
synthetic test suite with $260$ tasks for $256$ cores (and $p=32$
multicore nodes), to study the
scalability of our heuristics.  For all workloads, we modified speedup
profiles to satisfy Equations~\eqref{hyp.monotony} and~\eqref{hyp.total.work}.

As discussed in related work (see Section~\ref{sec.related}) and~\cite{syp},  
and confirmed by power measurement using Watts
Up Pro meters, we observed only minor power consumption variations of
less than 5\% when we limited co-scheduling to occur across multicore
nodes. 
Therefore, we only test configurations  where no more than a single application
can be scheduled on a given multicore node comprising 8 cores.
Adding a processor to an application $T_{i}$ which is already assigned $\sigma_{i}$ processors
actually means adding 8 new cores (a full
multicore node) to the $8 \sigma_{i}$ existing cores. Hence
a pack size of $k$ corresponds to the use of at most $8k$ cores for
applications in each pack. For Workloads-I and II, there are $16$ nodes and $128$ cores,
while Workload-III has up to $32$ nodes and $256$ cores.


\noindent
\textbf{Methodology for assessing the heuristics--}  
To evaluate the quality of the schedules generated by our heuristics,
we consider three measures: {\em Relative cost}, {\em Packing ratio}, and {\em Relative response time}.  
Recall that the cost of a pack is the maximum execution time of a task in that
pack and the cost of a co-schedule is the sum of the costs over all
its packs. 

We define the relative cost as the cost of a given co-schedule
divided by the cost of a 1-pack schedule, i.e., one with each task running
at maximum speed on all $p$ processors.  \\
For a given \kinp, consider  
$\sum_{i=1}^{n}{t_{i,\sigma(i)}\times \sigma(i)}$, i.e., the total work performed in
the co-schedule when the $i$-th task is assigned $\sigma(i)$ processors.
We define the packing ratio as this sum divided by $p$
times the cost of the co-schedule; observe that the packing quality is
high when this ratio is close to 1, meaning that there is almost no idle time
in the schedule. \\
An individual user could be
concerned about an increase in response time and a corresponding
degradation of individual productivity. To assess the impact on response
time, we consider the performance with respect to a relative
response time measure defined as follows. We consider a 1-pack schedule
with the $n$ tasks sorted in non-decreasing order of execution time, i.e., 
in a "shortest task first" order, to yield a minimal value of the response time.
If this ordering is given by the permutation $\pi(i), i=1,2,\ldots, n$,
the response time of task $i$ is $ r_i = \sum_{j=1}^{i}{t_{\pi(j),p}}$
and the mean response time is $R=\frac{1}{n}\sum_{i=1}^{n}{r_i}$.
For a given \kinp with $u$ packs 
scheduled in increasing order of the costs of a pack, the
response time of task $i$ in pack $v$, $1 \leq v \leq u$, assigned to $\sigma(i)$ processors, is:
$\hat{r}_i = \sum_{\ell=1}^{v-1}{cost(\ell)} + t_{i,\sigma(i)}$, where $cost(\ell)$ is
the cost of the $\ell$-th pack for $1\leq \ell \leq u$. The mean response time 
of the \kinp $\hat{R}$ is  calculated using these values
and we use $\frac{\hat{R}}{R}$ as the relative response time. 


\begin{figure*}[p]
\begin{center}
$\begin{array}{cc}
\includegraphics[width=.48\linewidth]{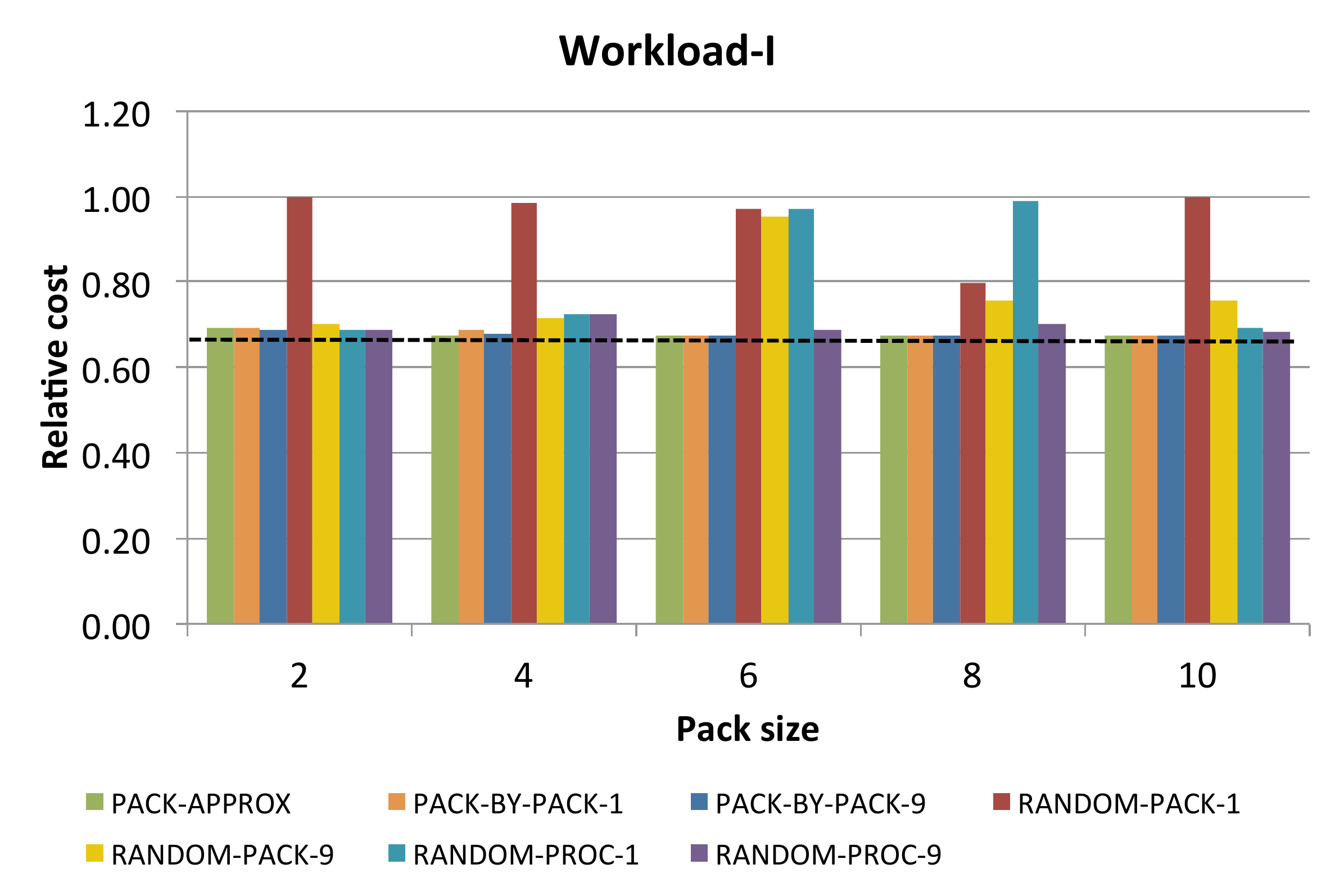} &
\includegraphics[width=.48\linewidth]{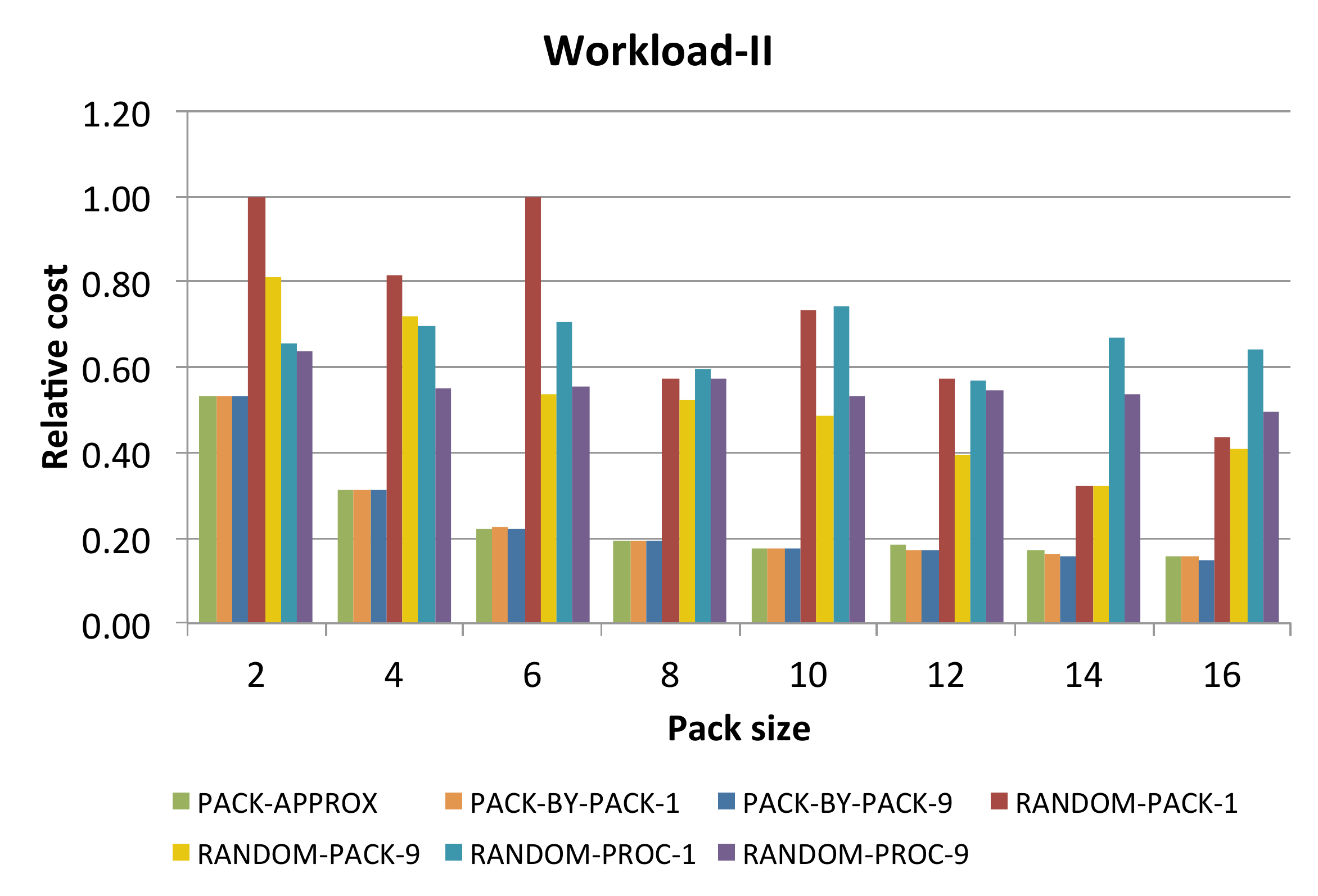} \\
(a) & (b)
\end{array}$
\end{center}
\caption{Quality of co-schedules: Relative costs 
are shown in (a) for Workload-I and in (b) for Workload-II.
The horizontal line in (a) indicates the relative cost of an optimal co-schedule for Workload-I.} 
\label{fig:relativecosts} 
\end{figure*}

\begin{figure*}[p]
\begin{center}
$\begin{array}{cc}
\includegraphics[width=.48\linewidth,clip]{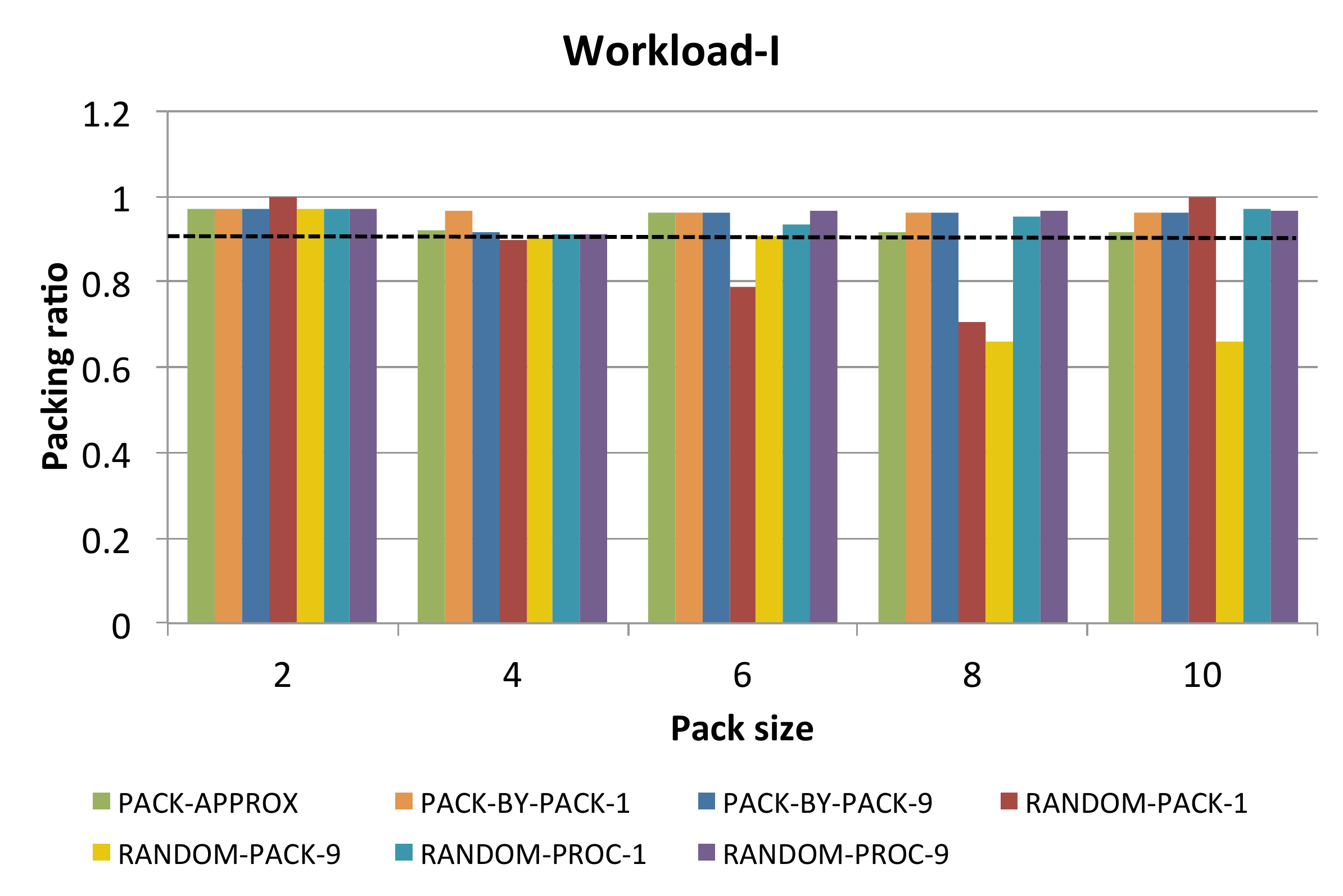} &
\includegraphics[width=.48\linewidth,clip]{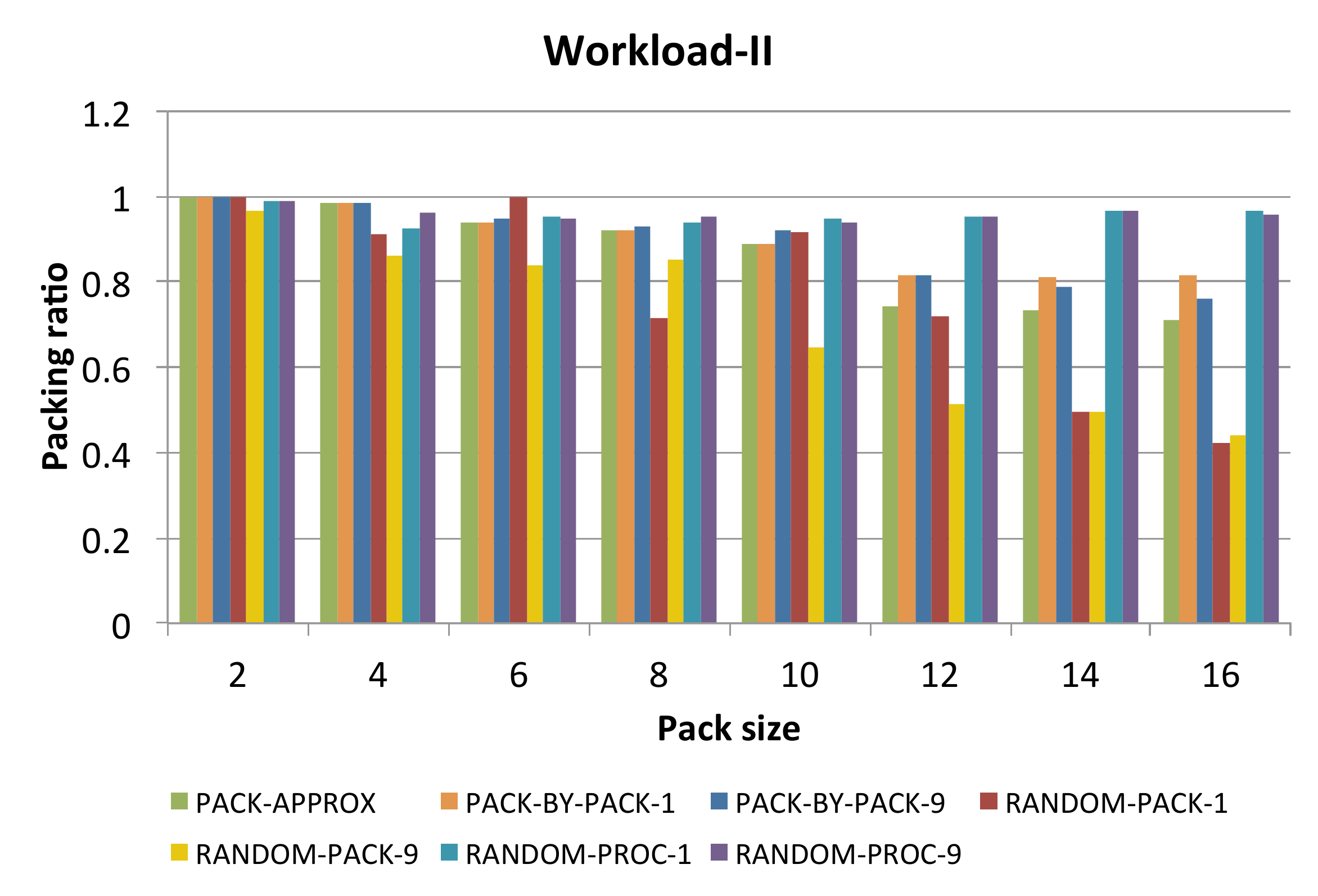} \\
(a) & (b)
\end{array}$
\end{center}
\caption{Quality of packs: Packing ratios  are shown in (a) for
Workload-I and in (b) for Workload-II.
The horizontal line in (a) indicates the packing ratio of an optimal co-schedule for Workload-I.} 
\label{fig:packingratio} 
\end{figure*}

\begin{figure*}[p]
\begin{center}
$\begin{array}{cc}
\includegraphics[width=.48\linewidth,clip]{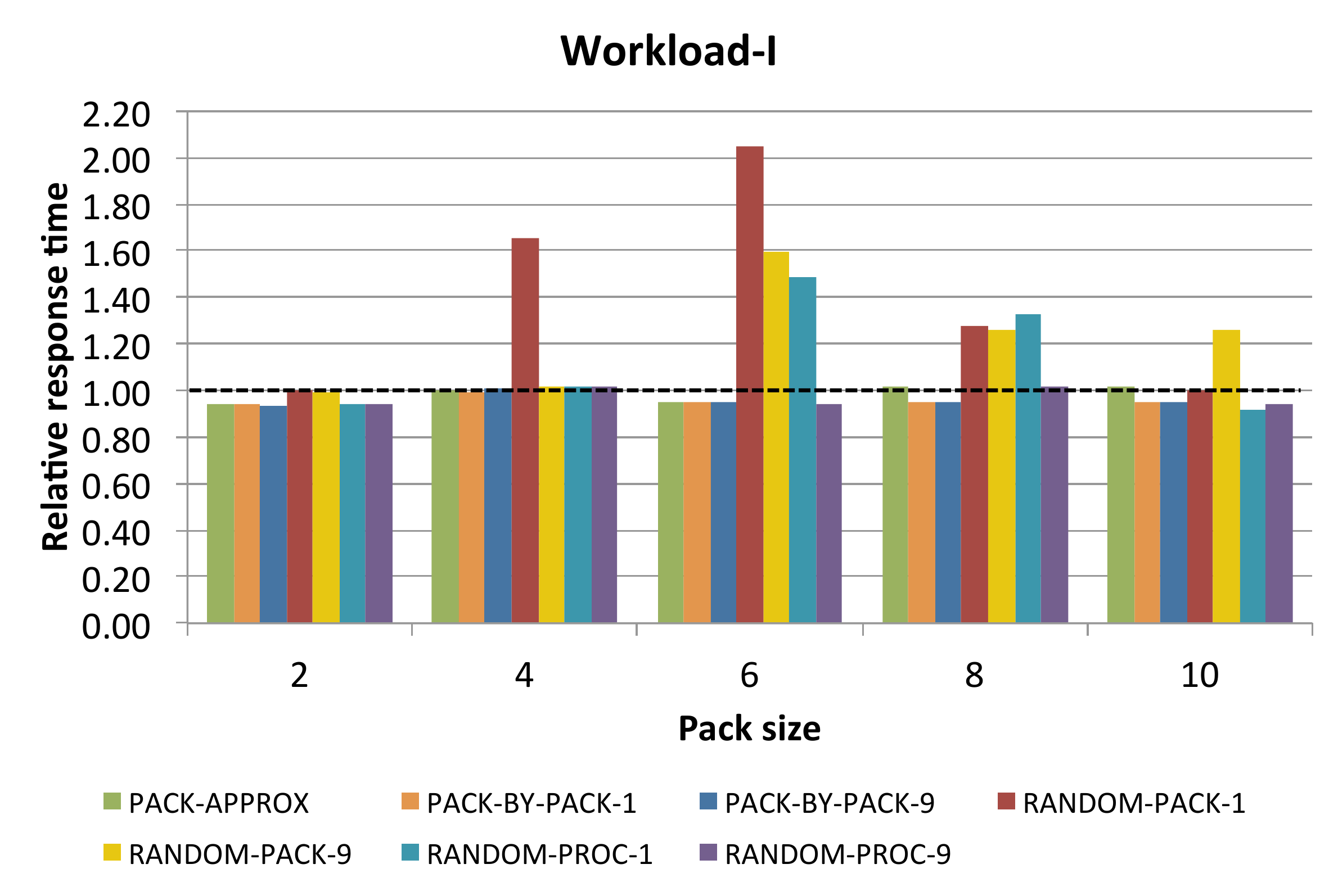} &
\includegraphics[width=.48\linewidth,clip]{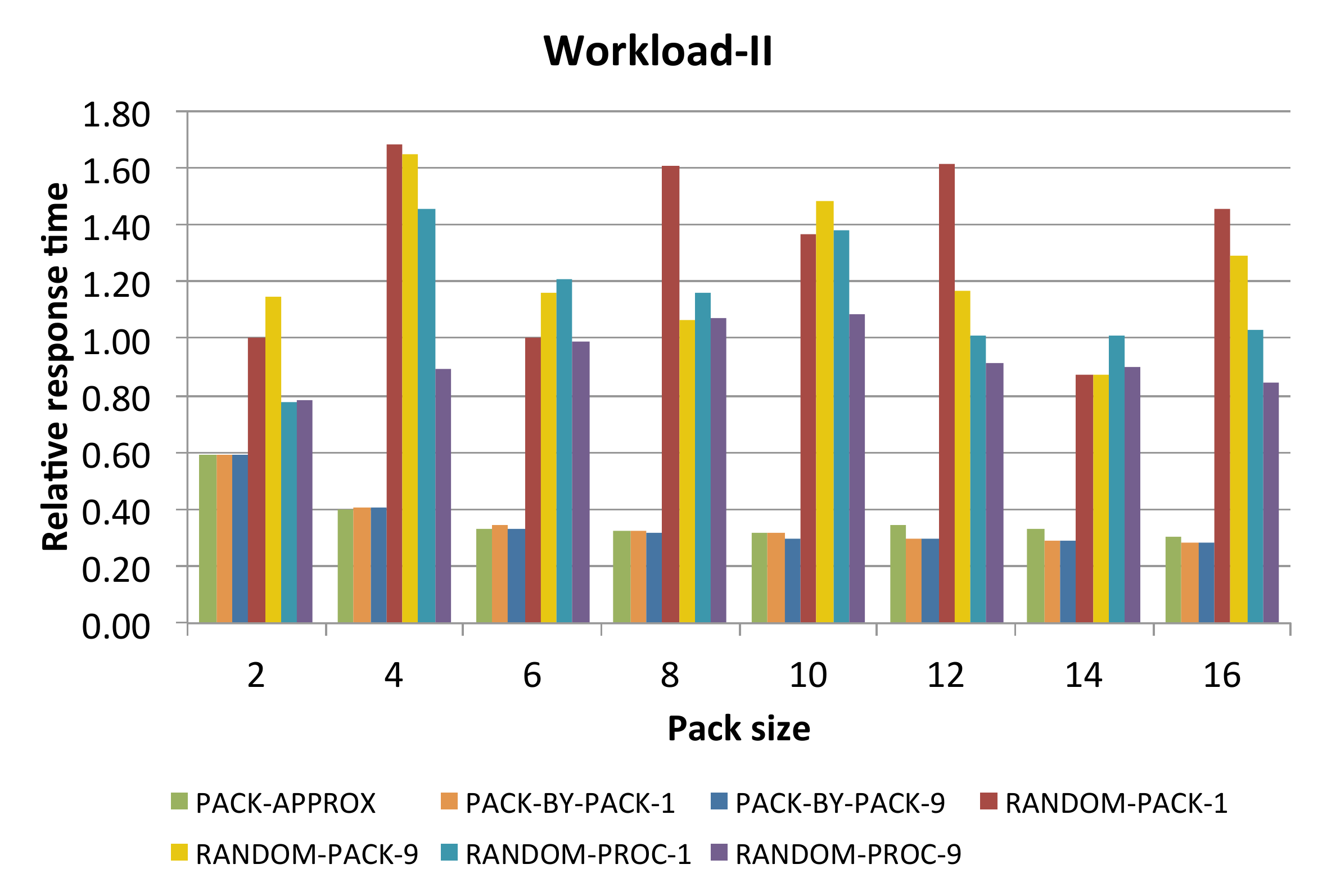}\\
(a) & (b) 
\end{array}$
\end{center}
\caption{Relative response times are shown in (a) for
Workload-I and in (b) for Workload-II; values less than 
1 indicate improvements in response times.
The horizontal line in (a) indicates the relative response time of an optimal co-schedule for Workload-I.} 
\label{fig:relativeresponsetimes} 
\end{figure*}

\begin{figure*}[t]
\begin{center}
$\begin{array}{ccc}
\includegraphics[width=.32\linewidth,clip]{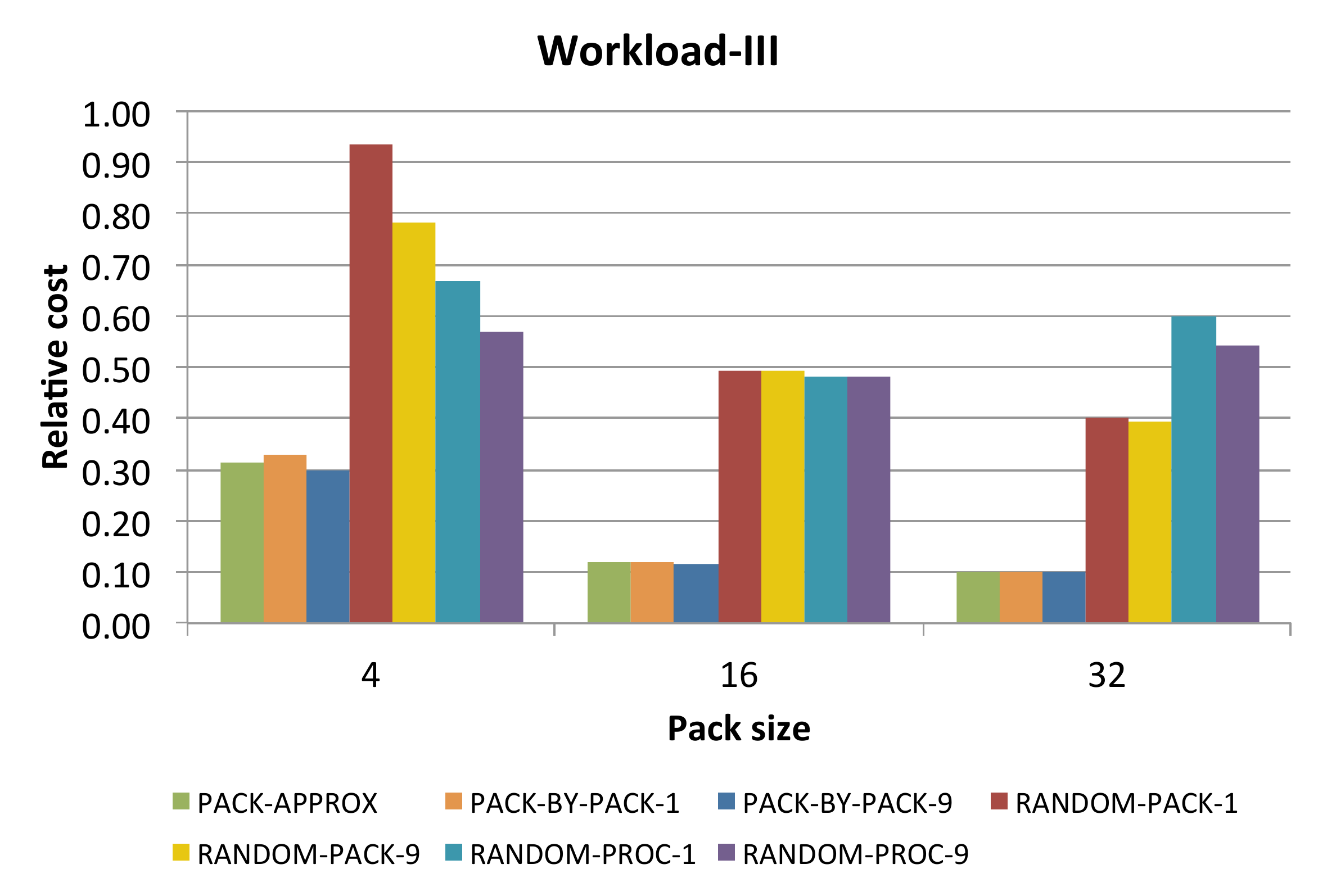} &
\includegraphics[width=.32\linewidth,clip]{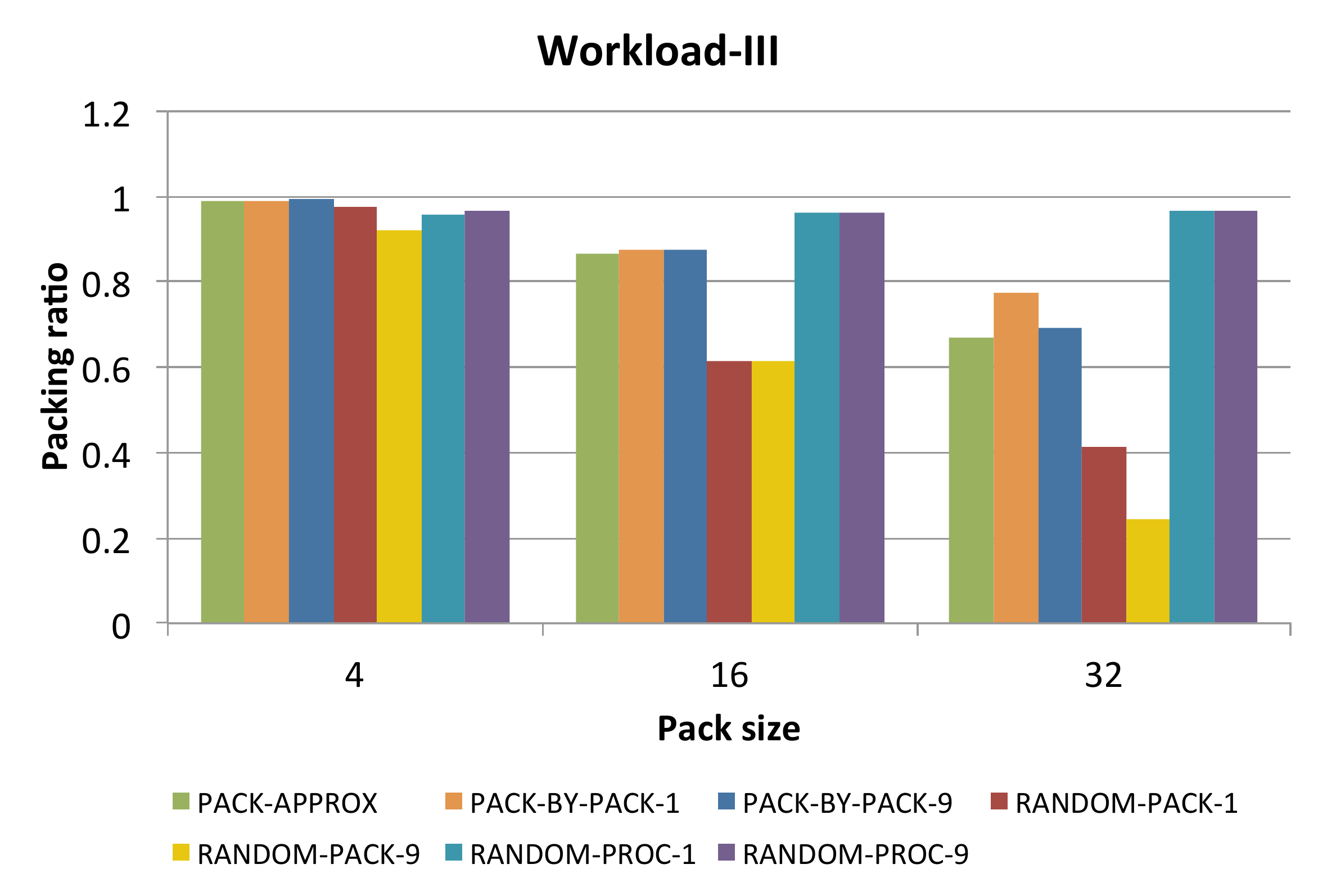} & 
\includegraphics[width=.32\linewidth,clip]{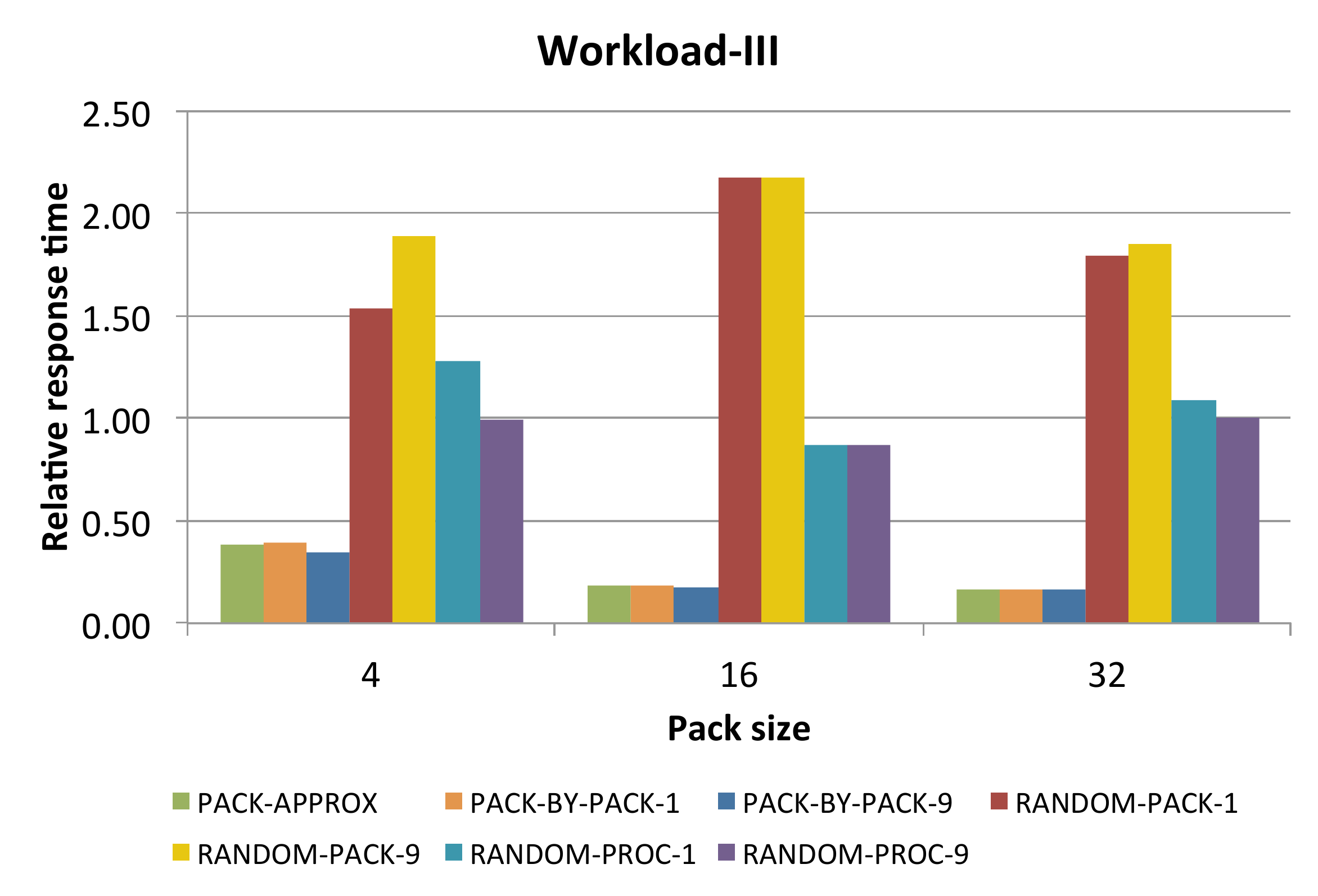} \\
(a) & (b) & (c)\\
\end{array}$
\end{center}
\caption{Relative costs, packing ratios and relative 
response times of co-schedules 
for Workload-III on $256$ cores.} 
\label{fig:workload-III} 
\end{figure*}

\noindent
\textbf{Results for small and medium workloads--}  
For Workload-I, we consider packs of size $k=2,4,6,8,10$ with $16$ processors
(hence a total of $128$ cores).  Note that we do not try $k=p=16$ since there
are only $10$ applications in this workload. For Workload-II, we consider packs of size $k=2,4,6,8,10,12,14,16$. 

Figure~\ref{fig:relativecosts} shows the relative cost of co-schedules
computed by the heuristics. 
For Workload-I (Figure~\ref{fig:relativecosts}(a)), the optimal co-schedule was constructed
using exhaustive search. 
We observe that the optimal co-schedule has
costs that are more than 35\% smaller than the cost of a 1-pack schedule
for Workload-I. Additionally, we observe that \ok and \bbb compute co-schedules
that are very close to the optimal one for all values of the pack size. 
Both \randompack and \randomproc
perform poorly when compared to \bbb and \ok, especially when
a single run is performed. 
As expected,
\randomproc does better than \randompack because it benefits from the
use of Algorithm~\ref{algo.makepack}, and for this small workload,
\randomproc-9 almost always succeed to find a near-optimal co-schedule. 
The results are similar
for the larger Workload-II as shown in Figure~\ref{fig:relativecosts}(b),
with an increased gap between random heuristics and the packing ones. 
Computing the optimal co-schedule was not feasible because of the 
exponential growth in running times for exhaustive search. 
With respect to the cost of a 1-pack schedule, we observe very significant
benefits, with a reduction in costs of most than  80\% for larger values
of the pack size, and in particular in the unconstrained case where $k=p=16$. 
This corresponds to significant savings in energy consumed
by the hardware for servicing a specific workload.

Figure~\ref{fig:packingratio} shows the quality of packing achieved
by the heuristics. The packing ratios are very close to one for
\bbb and \ok,  
indicating that our methods are producing high
quality packings.  In most cases, \randomproc and  \randompack
also lead to high packing ratios. 

Finally, Figure~\ref{fig:relativeresponsetimes} shows that \bbb and \ok 
produce lower cost schedules with commensurate reductions
in response times. For Workload-II and larger values of the pack size, response
time gains are over 80\%, making \kinp attractive from the
user perspective.


\noindent
\textbf{Scalability--} 
Figure~\ref{fig:workload-III} shows scalability trends 
for Workload-III with $260$ tasks on $32$ processors (hence a total of $256$ cores.) 
Although  
the heuristics, including \randompack and \randomproc,
 result in reducing costs relative to
those for a 1-pack schedule, \ok and \bbb are clearly superior, even
when the random schemes are run $9$ times. 
We observe that for pack sizes of $16$ and $32$,
\ok and \bbb produce
high quality co-schedules with costs and response times that
are respectively 90\% and 80\% lower than those for a 1-pack 
schedule. 
\bbb-1 obtains results that are very close to those of \bbb-9, 
hence even a single run returns a high quality co-schedule. 


\noindent
\textbf{Running times--} We report in Table~\ref{tab:time} the running times
of the seven heuristics. All heuristics run within a few milliseconds, even
for the largest workload. Note that \ok was faster on Workload-II than Workload-I
because its execution performed fewer iterations in this case. 
Random heuristics are slower than the other heuristics, because of the cost
of random number generation.
\bbb has comparable running times with \ok, even when $9$ values of $\varepsilon$ are used.

\begin{table}[h!]
\begin{center}
\begin{tabular}{|c|c|c|c|}
\hline
 &  Workload-I & Workload-II & Workload-III \\ \hline
 \ok & 0.50 & 0.30 & 5.12 \\
  \bbb-1 & 0.03 & 0.12 & 0.53 \\
  \bbb-9 & 0.30 & 1.17 & 5.07 \\
  \randompack-1 & 0.07 & 0.34 & 9.30 \\
  \randompack-9 & 0.67 & 2.71 & 87.25 \\
  \randomproc-1 & 0.05 & 0.26 & 4.49 \\
  \randomproc-9 & 0.47 & 2.26 & 39.54   \\\hline
\end{tabular}
\end{center}
\caption{Average running times in milliseconds.}
\label{tab:time}
\end{table}	

%

\noindent
\textbf{Summary of experimental results--}
Results  
indicate that heuristics  \ok and  \bbb   both
produce co-schedules of comparable quality.  
\bbb-9 is slightly better than \bbb-1, at a price of an increase in the running time from using more values of
$\varepsilon$. 
However, the running time remains very small, and similar to that of \ok. 
Using more values of $\varepsilon$ to improve \bbb  leads to small gains
in performance (e.g, $1\%$ gain for \bbb-9 compared to \bbb-1 for $k=16$ in Workload-II). However,
these small gains in performance correspond to significant gains in system throughput and energy, and
far outweigh the costs of computing multiple co-schedules. This 
makes  \bbb-9  the 
heuristic of choice.
Our experiments with $99$ values of $\varepsilon$ did not improve performance,  
indicating that large increases in  the number of  $\varepsilon$ values  may not be necessary.



\section{Conclusion}
\label{sec.conc}

We have developed and analyzed  
co-scheduling algorithms for
processing a workload of parallel tasks. Tasks are assigned to processors
and are partitioned into \packs of size $k$ with the constraint that
the total number of processors assigned over all tasks in a \pack does
not exceed $p$, the maximum number of available processors.  Tasks in
each \pack execute concurrently on a number of processors, and workload
completes in time equal to sum of the execution times of the \packs.
We have provided complexity results for 
minimizing the sum of the execution times of the \packs. The bad news is that 
this optimization problem
is NP-complete. This does not come as a surprise because we have to choose for each task both 
a number of processors and a \pack, and this double freedom induces a huge combinatorial solution space.
The good news is that we have provided an optimal resource allocation strategy once the packs are formed,
together with an efficient load-balancing algorithm to partition tasks with pre-assigned resources into packs.
This load-balancing algorithm is proven to be a $3$-approximation algorithm for the most general 
instance of the problem. Building upon these positive results, 
we have developed several heuristics that exhibit very good performance
in our test sets. These heuristics can significantly reduce the time for
completion of a workload for corresponding savings in system energy costs.
Additionally, these savings come along with measurable benefits in the
average response time for task completion, thus making it attractive
from the user's viewpoint.\\
These co-schedules can be computed very rapidly when speed-up profile
data are available.  Additionally, they operate at the scale of workloads
with a few to several hundred applications to deliver significant gains in
energy and time per workload.  These properties  present opportunities
for developing hybrid approaches that can additionally leverage
dynamic voltage and frequency scaling (DVFS) within an application.
For example, Rountree et al.~\cite{rountree} have shown that depending
on the properties of the application, DVFS can be applied at runtime
through their Adagio system, to yield system energy savings of  5\%
to 20\%.  A potential hybrid scheme could start with the computation 
of  a \kinp for a workload, following which DVFS could
be applied at runtime per application.\\ 
%
Our work indicates the potential benefits of co-schedules for high
performance computing installations where even medium-scale facilities
consume Megawatts of power. We plan to further test and extend this approach
towards deployment in university scale computing facilities where workload
attributes often do not vary much over weeks to months and energy costs
can be a  limiting factor.\\[.2cm]
{\em Acknowledgments.} This work was supported in part by
the ANR {\em RESCUE} project. The research of Padma Raghavan was supported in part by the 
The Pennsylvania State University and grants from the U.S. National Science Foundation.
The research of Manu Shantharam was supported in part by the U.S.  National Science Foundation 
award CCF-1018881.



\bibliographystyle{acm}
\bibliography{biblio,expbib,newrefs}

\begin{thebibliography}{10}

\bibitem{petsc}
{\sc Balay, S., Brown, J., Buschelman, K., Gropp, W.~D., Kaushik, D., Knepley,
  M.~G., McInnes, L.~C., Smith, B.~F., and Zhang, H.}
\newblock {PETSc} {W}eb page, 2012.
\newblock \url{http://www.mcs.anl.gov/petsc}.

\bibitem{sally}
{\sc Bhadauria, M., and McKee, S.~A.}
\newblock An approach to resource-aware co-scheduling for cmps.
\newblock In {\em Proceedings of the 24th ACM International Conference on
  Supercomputing\/} (New York, NY, USA, 2010), ICS '10, ACM, pp.~189--199.

\bibitem{abaqus}
{\sc Borgesson, L.}
\newblock Abaqus.
\newblock In {\em Coupled Thermo-Hydro-Mechanical Processes of Fractured Media
  - Mathematical and Experimental Studies}, O.~Stephansson, L.~Jing, and C.-F.
  Tsang, Eds., vol.~79 of {\em Developments in Geotechnical Engineering}.
  Elsevier, 1996, pp.~565 -- 570.

\bibitem{brucker1997scheduling}
{\sc Brucker, P., Gladky, A., Hoogeveen, H., Kovalyov, M.~Y., Potts, C.,
  Tautenhahn, T., and Van De~Velde, S.}
\newblock Scheduling a batching machine.
\newblock {\em Journal of Scheduling 1\/} (1998), 31--54.

\bibitem{chandra2005predicting}
{\sc Chandra, D., Guo, F., Kim, S., and Solihin, Y.}
\newblock Predicting inter-thread cache contention on a chip multi-processor
  architecture.
\newblock In {\em HPCA 11\/} (2005), IEEE, pp.~340--351.

\bibitem{coffman1980performance}
{\sc Coffman~Jr, E.~G., Garey, M.~R., Johnson, D.~S., and Tarjan, R.~E.}
\newblock Performance bounds for level-oriented two-dimensional packing
  algorithms.
\newblock {\em SIAM Journal on Computing 9}, 4 (1980), 808--826.

\bibitem{deb1973optimal}
{\sc Deb, R.~K., and Serfozo, R.~F.}
\newblock Optimal control of batch service queues.
\newblock {\em Advances in Applied Probability\/} (1973), 340--361.

\bibitem{dutot2003scheduling}
{\sc Dutot, P.-F., Mouni{\'e}, G., Trystram, D., et~al.}
\newblock Scheduling parallel tasks: Approximation algorithms.
\newblock {\em Handbook of Scheduling: Algorithms, Models, and Performance
  Analysis\/} (2003).

\bibitem{fcs}
{\sc Frachtenberg, E., Feitelson, D., Petrini, F., and Fernandez, J.}
\newblock Adaptive parallel job scheduling with flexible coscheduling.
\newblock {\em Parallel and Distributed Systems, IEEE Transactions on 16}, 11
  (2005), 1066--1077.

\bibitem{GareyJohnson}
{\sc Garey, M.~R., and Johnson, D.~S.}
\newblock {\em Computers and Intractability; A Guide to the Theory of
  NP-Completeness}.
\newblock W. H. Freeman \& Co., New York, NY, USA, 1990.

\bibitem{onchipcosched}
{\sc Hankendi, C., and Coskun, A.}
\newblock Reducing the energy cost of computing through efficient co-scheduling
  of parallel workloads.
\newblock In {\em Design, Automation Test in Europe Conference Exhibition
  (DATE), 2012\/} (2012), pp.~994--999.

\bibitem{shalf}
{\sc Kamil, S., Shalf, J., and Strohmaier, E.}
\newblock Power efficiency in high performance computing.
\newblock In {\em IPDPS\/} (2008), IEEE, pp.~1--8.

\bibitem{koole2001stochastic}
{\sc Koole, G., and Righter, R.}
\newblock A stochastic batching and scheduling problem.
\newblock {\em Probability in the Engineering and Informational Sciences 15},
  04 (2001), 465--479.

\bibitem{vasp}
{\sc Kresse, G., and Hafner, J.}
\newblock Ab initio molecular dynamics for liquid metals.
\newblock {\em Phys. Rev. B 47}, 1 (Jan 1993), 558--561.

\bibitem{li2010power}
{\sc Li, D., Nikolopoulos, D.~S., Cameron, K., de~Supinski, B.~R., and Schulz,
  M.}
\newblock Power-aware {M}{P}{I} task aggregation prediction for high-end
  computing systems.
\newblock In {\em IPDPS 10\/} (2010), pp.~1--12.

\bibitem{lodi2002two}
{\sc Lodi, A., Martello, S., and Monaci, M.}
\newblock Two-dimensional packing problems: A survey.
\newblock {\em European Journal of Operational Research 141}, 2 (2002),
  241--252.

\bibitem{lammps}
{\sc Plimpton, S.}
\newblock Fast parallel algorithms for short-range molecular dynamics.
\newblock {\em J. Comput. Phys. 117\/} (March 1995), 1--19.

\bibitem{potts2000scheduling}
{\sc Potts, C.~N., and Kovalyov, M.~Y.}
\newblock Scheduling with batching: a review.
\newblock {\em European Journal of Operational Research 120}, 2 (2000),
  228--249.

\bibitem{rountree}
{\sc Rountree, B., Lownenthal, D.~K., de~Supinski, B.~R., Schulz, M., Freeh,
  V.~W., and Bletsch, T.}
\newblock Adagio: making {D}{V}{S} practical for complex {H}{P}{C}
  applications.
\newblock In {\em ICS 09\/} (2009), pp.~460--469.

\bibitem{balaji}
{\sc Scogland, T., Subramaniam, B., and Feng, W.-c.}
\newblock {Emerging Trends on the Evolving Green500: Year Three}.
\newblock In {\em 7th Workshop on High-Performance, Power-Aware Computing\/}
  (Anchorage, Alaska, USA, May 2011).

\bibitem{syp}
{\sc Shantharam, M., Youn, Y., and Raghavan, P.}
\newblock Speedup-aware co-schedules for efficient workload management.
\newblock {\em Parallel Processing Letters\/} (2013, to appear).

\bibitem{turek1994scheduling}
{\sc Turek, J., Schwiegelshohn, U., Wolf, J.~L., and Yu, P.~S.}
\newblock Scheduling parallel tasks to minimize average response time.
\newblock In {\em Proceedings of the fifth annual ACM-SIAM symposium on
  Discrete algorithms\/} (1994), Society for Industrial and Applied
  Mathematics, pp.~112--121.

\end{thebibliography}

\end{document}